\documentclass[aps,prx,reprint,amsmath,amssymb,superscriptaddress,showkeys]{revtex4-1}

\usepackage[latin1]{inputenc}
\usepackage{graphicx}
\usepackage{amsmath}
\usepackage{amsfonts}
\usepackage{amsmath}
\usepackage{amssymb}
\usepackage{xcolor}
\usepackage{soul}
\usepackage{url}
\usepackage{bm}
\usepackage{amsthm}
\usepackage{mathtools}
\usepackage{dsfont}
\usepackage{stmaryrd}
\usepackage{adjustbox}
\usepackage{makecell}

\usepackage{epigraph}

\usepackage{cancel}

\let\emptyset\varnothing

\newcommand{\bs}[1]{\boldsymbol{#1}}

\renewcommand{\bs}[1]{\boldsymbol{#1}}

\newtheorem{definition}{Definition}

\newtheorem{lemma}{Lemma}
\newtheorem{proposition}{Proposition}
\newtheorem{corollary}{Corollary}
\newtheorem{example}{Example}

\usepackage{ctable}

\begin{document}

\title{Quantifying high-order interdependencies \\via multivariate extensions of the mutual information}

\author{Fernando Rosas}
\email{f.rosas@imperial.ac.uk}
\affiliation{Centre of Complexity Science and Department of Mathematics, Imperial College London, London SW7 2AZ, UK}
\affiliation{Department of Electrical and Electronic Engineering, Imperial College London, , London SW7 2AZ, UK}

\author{Pedro A.M. Mediano}
\affiliation{Department of Computing, Imperial College London, London SW7 2AZ, UK}

\author{Michael Gastpar}
\affiliation{School of Computer and Communication Sciences, EPFL, Lausanne 1015, Switzerland}

\author{Henrik J. Jensen}
\affiliation{Centre of Complexity Science and Department of Mathematics, Imperial College London, London SW7 2AZ, UK}
\affiliation{Institute of Innovative Research, Tokyo Institute of Technology, Yokohama 226-8502, Japan}

\date{\today}

\begin{abstract}

This article introduces a model-agnostic approach to study statistical synergy,
a form of emergence in which patterns at large scales are not traceable from
lower scales. Our framework leverages various multivariate extensions
of Shannon's mutual information, and introduces the \textit{O-information} as a
metric capable of characterising synergy- and redundancy-dominated systems. We
develop key analytical properties of the O-information, and study how it
relates to other metrics of high-order interactions from the statistical
mechanics and neuroscience literature. Finally, as a proof of concept, we
use the proposed framework to explore the relevance of statistical synergy in
Baroque music scores.

\end{abstract}

\maketitle

\section{Introduction}

A unique opportunity in the era of ``big data'' is to make use of the abundant
data to deepen our understanding of the high-order interdependencies that are
at the core of complex systems. Plentiful data is nowadays available about e.g.
the orchestrated activity of multiple brain areas, the relationship between
various econometric indices, or the interactions between different genes. What
allows these systems to be more than the sum of their parts is not in the
nature of the parts, but in the structure of their
interdependencies~\cite{crutchfield1994calculi}. However, quantifying the
``synergy'' within a given set of interdependencies is challenging, especially
in scenarios where the number of parts is far below the thermodynamic limit.

The relevance of synergistic relationships related to high-order interactions
has been thoughtfully demonstrated in the theoretical neuroscience literature.
For example, studies on neural coding have shown that neurons can carry
redundant, complementary or synergistic information -- the latter corresponding
to neurons that are uninformative individually but informative when considered
together~\cite{schneidman2003synergy,latham2005synergy}. Also, studies on
retina cells suggest that high-order Hamiltonians are necessary for
representing neurons firing in response to natural images, while pairwise
interactions suffice for neurons responding to less structured
stimuli~\cite{ganmor2011sparse}. Lastly, neuroimaging analyses have pointed out
the compatibility of local differentiation and global integration of different
brain areas, and suggested this to be a key capability for enabling high
cognitive functions~\cite{tononi1998complexity,balduzzi2008integrated}. Various
metrics have been proposed to distinguish these high-order features in data,
including the \textit{redundancy-synergy
index}~\cite{gat1999synergy,chechik2002group,varadan2006computational},
\textit{connected information}~\cite{schneidman2003network}, \textit{neural
complexity}~\cite{tononi1994measure}, and \textit{integrated
information}~\cite{barrett2011practical,mediano2018measuring}. While being
capable of capturing features of biological relevance, most of these metrics
have ad hoc definitions motivated by specific research agendas, and have few
theoretical guarantees~\footnote{An exception is the connected information,
which can be elegantly derived from principles of information
geometry~\cite{amari2001information}; however, there are no known methods to
compute this metric from data.}.

A promising approach for addressing high-order interdependencies is
\textit{partial information decomposition} (PID), which distinguishes different
``types'' of information that multiple predictors convey about a target
variable~\cite{williams2010nonnegative,griffith2014quantifying,wibral2017partial}.
In this framework, \textit{statistical synergies} are structures (or
relationships) that exist in the whole but cannot be seen in the parts, being
this rooted in the elementary fact that variables can be pairwise independent
while being globally correlated. Unfortunately, the adoption of PID has been
hindered by the lack of agreement on how to compute the components of the
decomposition, despite numerous recent
efforts~\cite{barrett2015exploration,ince2017measuring,james2018unique,finn2018pointwise}.
Moreover, the practical value of PID is greatly limited by the
super-exponential growth of terms for large systems, although some applications
do exist \cite{tax2017partial,wibral2017quantifying}.

The crux of multivariate interdependencies is that information-theoretic
descriptions of such phenomena are not straighforward, as extensions of
Shannon's classical results to general multivariate settings have proven
elusive~\cite{el2011network}. The most well-established multivariate extensions
of Shannon's mutual information are the \textit{total
correlation}~\cite{watanabe1960information} and the \textit{dual total
correlation}~\cite{han1975linear}, which provide suitable metrics of overall
correlation strength. Their values, however, differ in ways that are hard to
understand, even gaining the adjective of ``enigmatic'' among
scholars~\cite{james2011anatomy,vijayaraghavan2017anatomy}. Other popular
extension of the mutual information is the \textit{interaction
information}~\cite{mcgill1954multivariate}, which is a signed measure obtained
by applying the inclusion-exclusion principle to the Shannon
entropy~\cite{ting1962,yeung1991}. Although this metric provides insighful
results when applied to three variables, its is not easily interpretable when
applied to larger groups~\cite{williams2010nonnegative}.

This paper proposes to study multivariate interdependency via two dual
persectives: as \textit{shared randomness} and as \textit{collective
constraints} \footnote{This disctinction might not have been stressed in the
past because most studies focus on bivariate interactions between two sets of
variables, for which these two effects are equivalent and equal to the mutual
information. However, for interactions involving three or more variables these
perspectives differ.}. This setup leads to the \textit{O-information}, which --
following Occam's razor -- points out which of these perspectives provides a
more parsimonious description of the system. The O-information is found to
coincide with the interaction information for the case of three variables,
while providing a more meaningful extension for larger system sizes.

We show how the O-information captures the dominant characteristic of
multivariate interdependency, distinguishing redundancy-dominated scenarios
where three or more variables have copies of the same information, and
synergy-dominated systems characterised by high-order patterns that cannot be
traced from low-order marginals. In contrast with existing quantities that
require a division between predictors and target variables, the O-information
is -- to the best of our knowledge -- the first symmetric quantity that can
give account of intrinsic statistical synergy in systems of more than three
parts. Moreover, the computational complexity of the O-information scales
gracefully with system size, making it suitable for practical data analysis.

In the sequel, Section~\ref{sec:fundamentals} introduces the notions of shared
randomness and collective constraints, and Sections~\ref{sec:3} and \ref{sec:4}
present the O-information and its fundamental properties.
Section~\ref{sec:other_metrics} compares the O-information with other metrics
of high-order effects, and Section~\ref{sec:5} presents a case study on music
scores. Finally, Section~\ref{sec:discussion} summarises our main conclusions.

\section{Fundamentals}
\label{sec:fundamentals}

This section introduces two fundamental perspectives from which one can develop
an information-theoretic description of a system, and explains how they enable
novel perspectives to study interdependency.

\subsection{Entropy and negentropy}\label{sec:entro_neg}
 
\begin{quotation} 
\noindent\textit{For every outside there is an inside and for every inside there is an outside. And although they are different, they always go together.} 
\end{quotation} 
\begin{flushright} 
Alan Watts, \textit{Myth of myself}
\end{flushright}

Following the Bayesian interpretation of information theory, we define the
\textit{information contained in a system} as the amount of data that an
observer would gain after determining its configuration -- i.e. after measuring
it~\cite{jaynes2003probability}. If each possible configuration is to be
represented by a distinct sequence of bits, source coding theory~\cite[Ch.
5]{cover2012elements} shows that an optimal (i.e. shortest) labelling depends
on prior information available before the measurement.
Information, hence, refers to how the state of knowledge of the observer
changes after the system is measured, quantifying the amount of bits that are
revealed through this process~\footnote{For a quantum-mechanical treatment of
this notion, see \cite[Ch. 2]{breuer2002theory}}.

Let us consider an observer measuring a system composed by
$n$ discrete variables, $\bm{X}^n=(X_1,\dots,X_n)$. If the observer only knows
that each variable $X_j$ can take values over a finite alphabet $\mathcal{X}_j$
of cardinality $|\mathcal{X}_j|$, the amount of information needed to specify
the state of $X_j$ is $\log |\mathcal{X}_j|$ 
(logarithms are calculated using base $2$ unless specified otherwise).
In contrast, if the observer knows that the system's behaviour follows a
probability distribution $p_{\bm{X}^n}$, then the average amount of information
in the system reduces to the \textit{entropy} $H(\bm{X}^n) \coloneqq
-\sum_{\bm{x}^n} p_{\bs{X}^n}(\bm x^n) \log p_{\bs{X}^n}(\bm x^n)
$~\cite{jaynes2003probability}. The difference
\begin{equation}
  \mathcal{N}(\bs{X}^n) \coloneqq \sum_{j=1}^n \log |\mathcal{X}_j| - H(\bs{X}^n) 
  \label{eq:negeee}%
\end{equation}
is known as \emph{negentropy}~\cite{brillouin1953negentropy}, and corresponds
to the information about the system that is disclosed by its statistics, before
any measurement takes place.

Probability distributions are, from this perspective, a compendium of soft and
hard constraints that reduce the effective phase space that the system can
explore -- hard constraints completely forbid some configurations, while soft
constraints make them improbable. Consequently, a given distribution divides
the phase space in an admisible region quantified by the entropy, and an
inadmissible region quantified by the negentropy~\footnote{This observation can
be made rigorous via the Shannon-McMillan-Breiman theorem~\cite[Sec.
3]{cover2012elements}.}. Each part describes the system's structure from a
different point of view: the entropy refers to what the system can do, while
the negentropy refers to what it can't.

\subsection{The two faces of interdependency}
\label{sec:two_faces}

\subsubsection{Collective constraints}

\begin{figure*}
 \centering
 \includegraphics[width=6.5in]{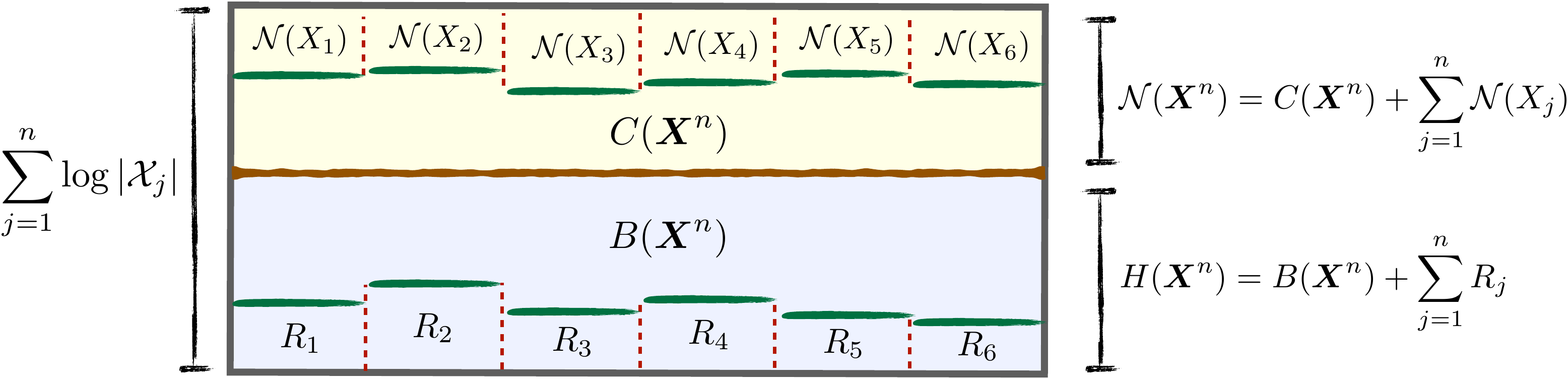}
 \caption{Following Eq.~\eqref{eq:info_dec}, the total information that can be stored in the    
   system $\bm X^n$ ($\sum_{j=1}^n \log | \mathcal{X}_j|$) is divided by a given state of knowledge (i.e. a
   probability distribution) into what is determined by the constraints
   ($\mathcal{N}(\bs{X}^n)$) and what is not instantiated until an actual
   measurement takes place ($H(\bs{X}^n)$). Moreover, both terms can be further
   decomposed into their individual and collective components, yielding different
   perspectives on interdependency seen as either collective constraints
   ($C(\bs{X}^n)$) or shared randomness ($B(\bs{X}^n)$).}
 \label{fig:diag1}
\end{figure*}

In the same way as $\mathcal{N}(\bs{X}^n)$ quantifies the strength of the
overall constraints that rule the system, the constraints that affect
individual variables are captured by the \emph{marginal negentropies}
$\mathcal{N}(X_j) \coloneqq \log |\mathcal{X}_j| - H(X_j)$. Intuitively, the
constraints that affect the whole system are richer than individual
constraints, as the latter do not take into account collective effects. Their
difference,
\begin{equation}\label{eq:TC}
  \begin{split}
    C(\bs{X}^n) \coloneqq& \; \mathcal{N}(\bs{X}^n) - \sum_{j=1}^n \mathcal{N}(X_j) \\
      =& \sum_{j=1}^n H(X_j) - H(\bs{X}^n) ~ ,
  \end{split}
\end{equation}
quantifies the strength of the ``collective constraints.'' This quantity is
known as \textit{total correlation}~\cite{watanabe1960information} (or
\textit{multi-information}~\cite{studeny1998multiinformation}). By re-writing
this relationship as $\mathcal{N}(\bs{X}^n) = \sum_j \mathcal{N}(X_j) +
C(\bs{X}^n)$ one finds that the constraints prescribed by the distribution are
of two types: constraints confined to individual variables, and collective
constraints that restrict groups of two or more variables.

\begin{example}\label{ex:2}
Consider $X_1$ and $X_2$ to be binary random variables with $p_{X_1, X_2}(0,1)
= p_{X_1, X_2}(1,0) = 1/2$. This distribution divides the total information
(two bits) into $H(X_1, X_2) = 1$ and $\mathcal{N}(X_1, X_2) = 1$ . Moreover,
$\mathcal{N}(X_1) = \mathcal{N}(X_2) = 0$ and therefore $C(X_1, X_2) =
\mathcal{N}(X_1, X_2) = 1$, confirming that the constraints act on both $X_1$
and $X_2$.

As a contrast, consider $Y_1$ and $Y_2$ binary random variables with
distribution $p_{Y_1,Y_2}(0,0) = p_{Y_1,Y_2}(1,0) = 1/2$. In this case
$\mathcal{N}(Y_1)=0$ while $\mathcal{N}(Y_2) = \mathcal{N}(Y_1, Y_2) = 1$,
showing that the only constraint in this system acts solely over $Y_2$.
Accordingly, for this case $C(Y_1, Y_2) = 0$.
\end{example}

\subsubsection{Shared randomness}\label{sec:shared_randomness}

As we did for $\mathcal{N}(\bm X^n)$, let us decompose $H(\bm X^n)$ in
individual and collective components. To do this, we introduce the quantity
$R_j = H(X_j|\bs{X}_{-j}^n)$ as a metric of how independent $X_j$ is from the
rest of the system $\bs{X}^n_{-j} = (X_1,\dots,X_{j-1},X_{j+1},\dots,X_n)$.
According to distributed source coding theory~\cite[Ch. 10.5]{el2011network},
$R_j$ corresponds to the data contained in $X_j$ that cannot be extracted from
measurements of other variables~\footnote{In fact, a direct calculation shows
that the variables $\bs{X}^n$ are independent if and only if
\unexpanded{$\sum_j R_j=H(\bs{X}^n)$}.}. The quantity $\sum_{j=1}^n R_j$ is
known as the \textit{residual entropy}~\cite{abdallah2012measure} (originally
introduced under the name of \textit{erasure
entropy}~\cite{verdu2006erasure,verdu2008information}), and quantifies the
total information that can only be accessed by measuring a specific variable,
i.e. the amount of ``non-shared randomness.'' Accordingly, the difference
\begin{equation}\label{eq:DTC}
B(\bs{X}^n) \coloneqq H(\bs{X}^n) - \sum_{j=1}^n R_j
\end{equation}
quantifies the amount of information that is shared by two or more variables --
equivalently, information that can be accessed by measuring more than one
variable. Although this quantity was introduced under the name of \textit{dual
total correlation}~\cite{han1975linear} (also known as \textit{excess entropy}
\cite{olbrich2008should} or \textit{binding information}
\cite{abdallah2012measure,vijayaraghavan2017anatomy}), we prefer the name
\textit{binding entropy} as it emphasises the fact that it is actually a part
of the entropy. As the entropy corresponds to the randomness within the system,
the binding entropy quantifies the ``shared randomness'' that exists among the
variables.

\begin{example}
Let us consider $X_1,X_2$ and $Y_1,Y_2$ from Example~\ref{ex:2}. For the former
system one finds that $R_1=R_2=0$ and hence $B(X_1,X_2) = H(X_1,X_2) = 1$,
which means that the randomness within the system can be retrieved from
measuring either $X_1$ or $X_2$. In contrast, when considering $Y_1,Y_2$ one
finds that $R_2 = 0$ and $R_1 = H(Y_1,Y_2) = 1$, and hence $B(Y_1,Y_2) = 0$.
This implies that the randomness of the system can be retrieved by measuring
only $Y_1$.
\end{example}

Wrapping up, one can re-write Eq.~\eqref{eq:negeee} using Eqs.~\eqref{eq:TC}
and \eqref{eq:DTC} and express the total information encoded in the system
described by $\bs{X}^n$ in terms of constraints and randomness:
\begin{equation}\label{eq:info_dec}
  \begin{split}
\sum_{j=1}^n &\log | \mathcal{X}_j| = \; \mathcal{N}(\bs{X}^n) + H(\bs{X}^n) \\
=& \underbrace{\left[ C(\bs{X}^n) + \sum_{j=1}^n \mathcal{N}(X_j) \right]}_{\mathclap{\substack{\mathrm{Collective~and~individual}\\[0.5ex]\mathrm{constraints}}}} + \underbrace{\left[ B(\bs{X}^n) + \sum_{j=1}^n R_j \right]}_{\mathclap{\substack{\mathrm{Shared~and~private}\\[0.5ex]\mathrm{randomness}}}}. 
  \end{split}
\end{equation}
This decomposition is illustrated in Figure~\ref{fig:diag1}.
%

\section{Introducing the O-information}
\label{sec:3}

\subsection{Definition and basic properties}

The total correlation and the binding entropy provide complementary metrics of
interdependence strength. Following Occam's Razor, one might ask which of these
perspectives allows for a shorter (i.e. more parsimonious) description. This is
answered by the following definition:

\begin{definition}\label{def:omega}
The \textit{O-information} (shorthand for ``information about Organisational
structure'') of the system described by the random vector $\bs{X}^n$ is defined
as
\begin{align}
\Omega(\bs{X}^n) \!\coloneqq& \; C(\bs{X}^n) - B(\bs{X}^n) \\
=& (n-2) H(\bs{X}^n)
+ \sum_{j=1}^n \big[ H(X_j) -  H(\bs{X}_{-j}^n) \big] . \nonumber
\end{align}
\end{definition}

Intuitively, $\Omega(\bs{X}^n) > 0$ states that the interdependencies can be
more efficiently explained as shared randomness, while $\Omega(\bs{X}^n) < 0$
implies that viewing them as collective constraints can be more convenient.
Note that $\Omega(\bs{X}^n)$ was first introduced as ``enigmatic information''
in Ref.~\cite{james2011anatomy}, although now that its properties have been
revealed we choose to give it a more appropriate name.

To develop some insight about the O-information, let us compare it with the
\textit{interaction information}~\footnote{The interaction information is
closely related to the \textit{I-measures}~\cite{yeung1991}, the
\textit{co-information}~\cite{Bell2003}, and the \textit{multi-scale
complexity}~\cite{bar2004multiscale}.}, which is a signed metric defined by
\begin{equation}\label{eq:interaction}
I(X_1;X_2;\dots;X_n) \coloneqq - \sum_{\bs{\gamma}\subseteq\{1,\dots,n\}} (-1)^{|\bs{\gamma}|} H(\bs{X}^{\bs{\gamma}})~,
\end{equation}
where the sum is performed over all subsets
$\bs{\gamma}\subseteq\{1,\dots,n\}$, with $|\bs{\gamma}|$ being the cardinality
of $\bs{\gamma}$ and $\bs{X}^{\bs{\gamma}}$ the vector of all variables with
indices in $\bs{\gamma}$. For $n=2$, Eq.~\eqref{eq:interaction} reduces to the
well-known \textit{mutual information},
\begin{equation}
I(X_1;X_2) = H(X_1) + H(X_2) - H(X_1,X_2)~. \nonumber
\end{equation}
For $n=3$, Eq.~\eqref{eq:interaction} gives
\begin{align}\label{eq:red-syn}
I(X_1;X_2;X_3) \! &= I(X_i;X_j) - I(X_i;X_j|X_k) \\
&= I(X_i;X_j) + I(X_i;X_k) - I(X_i;X_j,X_k) \nonumber
\end{align}
for $\{i,j,k\} = \{1,2,3\}$, which is known to measure the difference between
synergy and redundancy~\cite{williams2010nonnegative}. Specifically, redundancy
dominates when $I(X_1;X_2;X_3) \geq 0$; e.g. if $X_1$ is a Bernoulli random
variable with $p=1/2$ and $X_1=X_2=X_3$, then $I(X_1;X_2;X_3)=1$. In contrast,
synergy dominates when $I(X_1;X_2;X_3) \leq 0$, corresponding to statistical
structures that are present in the full distribution but not in the pairwise
marginals. For example, if $Y_1$ and $Y_2$ are independent Bernoulli variables
with $p=1/2$ and $Y_3 = Y_1 + Y_2 \pmod{2}$ (i.e. an \texttt{xor} logic gate)
then $I(Y_1;Y_2;Y_3) = -1$, since these variables are pairwise independent
while globally correlated~\cite{rosas2016understanding}. Unfortunately, for $n
\geq 4$ the co-information no longer reflects the balance between redundancy
and synergy~\cite[Section V]{williams2010nonnegative}.

To contrast with the interaction information, the next Lemma presents some
basic properties of $\Omega$ (the proofs are left for the reader).
\begin{lemma}\label{prop:basic}
The O-information satisfies the following properties:
\begin{itemize}
\item[(A)] $\Omega$ does not depend on the order of $X_1,\dots,X_n$.
\item[(B)] $\Omega(X_1,X_2) = 0$ for any $p_{X_1 X_2}$. 
\item[(C)] $\Omega(X_1,X_2,X_3) = I(X_1;X_2;X_3)$ for any $p_{\bm X^3}$. 
\end{itemize}
\end{lemma}

Property (A) shows that $\Omega$ reflects an intrinsic property of the system,
without the need of dividing the variables in groups with differentiated roles
(e.g. targets vs predictors, or input vs output). Property (B) confirms that
$\Omega$ captures only interactions that go beyond pairwise relationships.
Finally, Property (C) shows that when $n=3$ the O-information is equal to
$I(X_1;X_2;X_3)$. Interestingly, a direct calculation shows that if $n>3$ then
in general $\Omega(\bs{X}^n) \neq I(X_1;X_2;\dots;X_n)$.

At this stage, one might wonder if the O-information could provide a metric for
quantifying the balance of redundancy and synergy, as the interaction
information does for $n=3$. Intutively, one could expect redundant systems to
have small $B(\bs{X}^n)$ due to the multiple copies of the same information
that exist in the system, while having large values of $C(\bs{X}^n)$ because of
the constraints that are needed to ensure that the variables remain correlated.
On the other hand, synergistic systems are expected to have small values of
$C(\bs{X}^n)$ due to the few high-order constraints that rule the system, while
having larger values of $B(\bs{X}^n)$ due to the weak low-order structure.
These insights are captured in the following definition, which is supported by
multiple findings presented in the following sections.

\begin{definition}
If $\Omega(\bs{X}^n) > 0$ we say that the system is
\textit{redudancy-dominated}, while if $\Omega(\bs{X}^n) < 0$ we say it is
\textit{synergy-dominated}.
\end{definition}

In previous work we used another metric to assess synergy- and
redundancy-dominated systems~\cite{rosas2018selforg}. Appendix
\ref{app:selforg} provides an analytical and numerical account of the
consistency between these two metrics.

\subsection{Information decompositions}

\subsubsection{The lattice of partitions}
\label{sec:decompositions}\

\begin{figure*}
\begin{center}
\begin{adjustbox}{width=14cm}
  \includegraphics{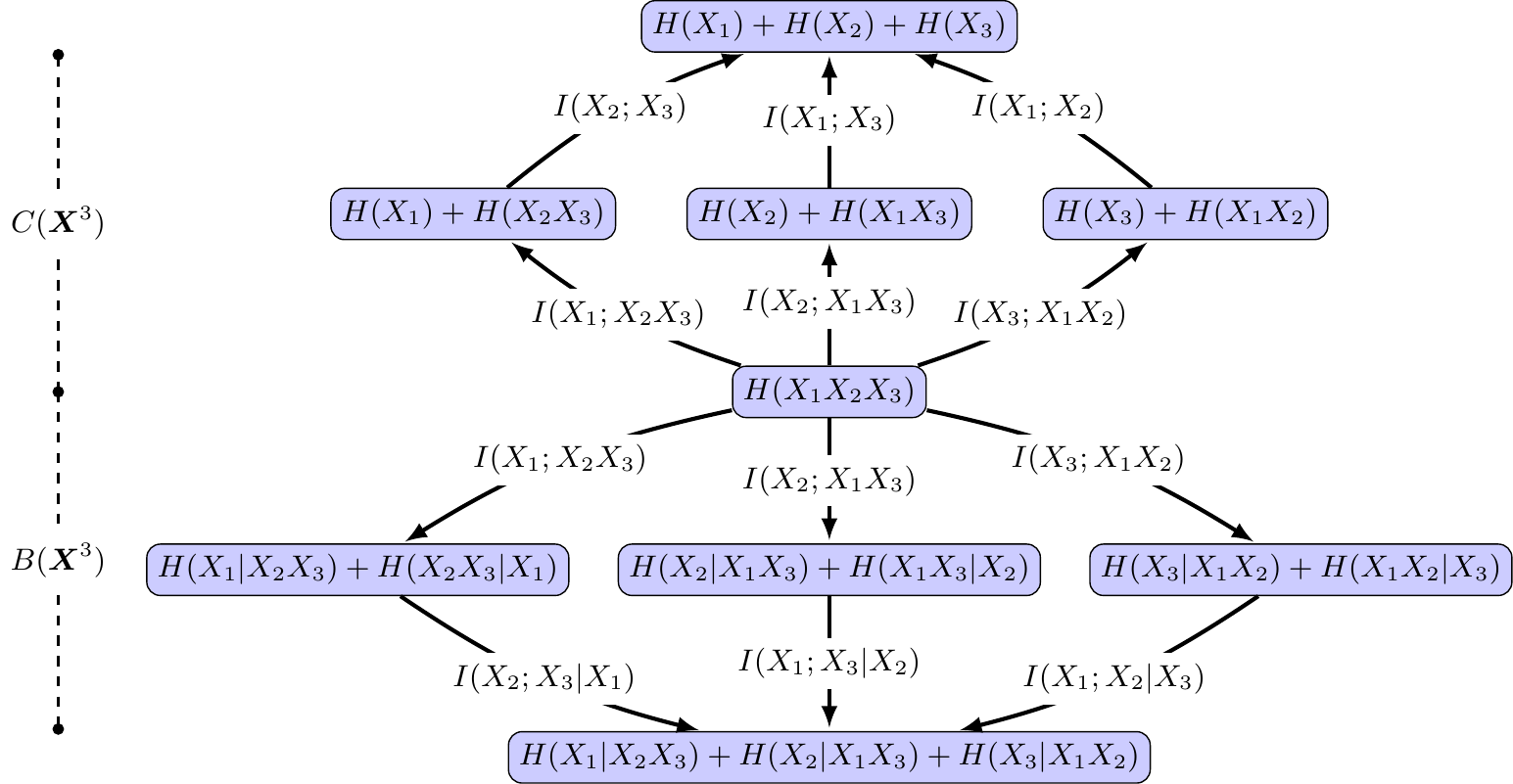}
\end{adjustbox}
\end{center}
\caption{\textit{Double diamond diagram} with the possible sequences of binary partitions of three variables.
Every path from the source node ($H(\bm X^3$) to the two sink nodes ($H(X_1)+H(X_2)+H(X_3)$ and $H(X_1|X_2X_3)+H(X_2|X_1X_3)+H(X_3|X_1X_2)$) corresponds to a decomposition of either
$C(\bs{X}^3)$ or $B(\bs{X}^3)$.}
\label{fig:diagramD3}
\end{figure*}

A partition $\pi = (\bs{\alpha}_1|\bs{\alpha}_2 |\dots| \bs{\alpha}_m) $ of the
indices $\{1,\dots,n\}$ is a collection of \textit{cells}
$\bs{\alpha}_j=\{\alpha_j^1,\dots,\alpha_j^{{l(j)}}\} $ that are disjoint and
satisfy $\bigcup_{j=1}^m\bs{\alpha}_j = \{1,\dots,n\}$. The collection of all
possible partitions of $\{1,\dots,n\}$, denoted by $\mathcal{P}_n$, has a
lattice structure~\footnote{A lattice is a partially ordered set with a unique
infimum and supremum. For more details on this construction, see
\cite{stanley2012}} enabled by the partial order introduced by the natural
refinement relationship, in which $\pi_2\succeq\pi_1$ if $\pi_2$ is
\textit{finer}~\footnote{If $\pi_1,\pi_2\in\mathcal{P}_n$ with $\pi_1=
(\bs{\alpha}_1|\dots|\bs{\alpha}_r)$ and $\pi_2 =
(\bs{\beta}_1|\dots|\bs{\beta}_s)$, $\pi_1$ is \textit{finer} than $\pi_2$ if
for each $\bs{\alpha}_i$ exists $\bs{\beta}_k$ such that $\bs{\alpha}_i \subset
\bs{\beta}_k$.} than $\pi_1$ (or, equivalently, if $\pi_1$ is \textit{coarser}
than $\pi_2$). A partition $\pi_2$ is said to \textit{cover} $\pi_1$ if
$\pi_2\succeq\pi_1$ and it is not possible to find another partition $\pi_3$ such
that $\pi_2\succeq\pi_3\succeq\pi_1$~\footnote{It is direct to see that $\pi_2$
covers $\pi_1$ if and only if it is an ``elementary refinement'', i.e. $\pi_2$
can be obtained from $\pi_1$ by dividing one cell of $\pi_1$ in two. Hence, if
$\pi_2$ covers $\pi_1$ then $|\pi_2|=|\pi_1|+1$, where $|\pi|$ is the number of
(non-empty) cells of $\pi$.}. For this partial order relationship,
$\pi_\text{source} = (12\dots n)$ is the unique infimum of $\mathcal{P}_n$, and
$\pi_\text{sink} = (1|2|\dots |n)$ is the unique supremum of $\mathcal{P}_n$.

A directed acyclic graph (DAG) $\mathcal{G}_n$ can be built, where the nodes
are the partitions in $\mathcal{P}_n$, and a directed edge exists from $\pi_1$
to $\pi_2$ if and only if $\pi_2$ covers $\pi_1$~\footnote{Put simply, there is an
edge from $\pi_1$ to $\pi_2$ if $\pi_2$ results from taking $\pi_1$ and
splitting one of its cells in two.}. A path $\texttt{p}$ in $\mathcal{G}_n$
joining two partitions $\pi_\text{a}$ and $\pi_\text{b}$ is a sequence of nodes
$\texttt{p} = (\pi_1,\dots,\pi_L)$, where $\pi_1=\pi_\text{a}$,
$\pi_L=\pi_\text{b}$, and $\pi_{i+1}$ covers $\pi_{i}$ for all
$i\in\{1,\dots,L-1\}$. The collection of all paths from $\pi_\text{a}$ to
$\pi_\text{b}$ is denoted by $\texttt{P}(\pi_\text{a},
\pi_\text{b})$~\footnote{It is direct to check that $\pi_\text{b} \succ
\pi_\text{a}$ if and only if $\texttt{P}(\pi_\text{a}, \pi_\text{b})\neq
\emptyset$. Moreover, all $\texttt{p}\in\texttt{P}(\pi_\text{a},\pi_\text{b})$
have the same length, given by $|\texttt{p}| = | |\pi_\text{b}| -
|\pi_\text{a}| | $, where $|\texttt{p}|$ is the number of edges in the path.}.
If the edge joining $\pi_1$ and $\pi_2$ has a weight $v(\pi_1,\pi_2)$
associated, then the corresponding \textit{path weight} of $\texttt{p} =
(\pi_1,\dots,\pi_L)$ is merely the summation of all edge weights along
\texttt{p}:
\begin{equation}\label{eq:def_pathweight1}
  W(\texttt{p};v) := \sum_{k=1}^{L-1} v(\pi_k,\pi_{k+1}) ~ .
\end{equation}

\subsubsection{Lattice decompositions of $C(\bs{X}^n)$ and $B(\bs{X}^n)$}

Let us build some useful weight functions over $\mathcal{G}_n$. We first assign
to each node $\pi = (\bs{\alpha}_1|\dots|\bs{\alpha}_L) \in \mathcal{P}_n$ the
value
\begin{equation}
H(\pi) \coloneqq H\big( \prod_{j=1}^L p_{ \bs{X}^{\bs{\alpha}_j}} \big) = \sum_{j=1}^L H\big( \bs{X}^{\bs{\alpha}_j} \big) \nonumber
\enspace,
\end{equation}
with $\bs{X}^{\bs{\alpha}_j}=(X_{\alpha_j^1},\dots, X_{\alpha_j^{l(j)}})$,
which corresponds to the entropy of the probability distribution $\prod_{j=1}^L
p_{ \bs{X}^{\bs{\alpha}_j}}$ that includes interdependencies within cells, but
not across cells. To each edge of $\mathcal{G}_n$ we assign a weight
\begin{equation}\label{weight_h}
v_\text{h}(\pi_1,\pi_2) \coloneqq H(\pi_2) - H(\pi_1)
\enspace.
\end{equation}
Since $H(\pi_\text{a}) \geq H(\pi_\text{b})$ if $\pi_\text{a} \succeq
\pi_\text{b}$, one can represent $\mathcal{G}_n$ under $v_\text{h}$ by placing
nodes with more cells in higher layers (see the upper half of
Figure~\ref{fig:diagramD3}).

Alternatively, let us now consider the residual entropy of $\pi =
(\bs{\alpha}_1|\dots|\bs{\alpha}_m)\in\mathcal{P}_n$, which is given by $R(\pi)
\coloneqq \sum_{k=1}^m R_{\bs{\alpha}_k}$, with
\begin{equation}\label{def:red_part}
R_{\bs{\alpha}_k} \coloneqq H(\bs{X}^{\bs{\alpha}_k}|\bs{X}^{\bs{\alpha}_1},\dots, \bs{X}^{\bs{\alpha}_{k-1}}, \bs{X}^{\bs{\alpha}_{k+1}},\dots,\bs{X}^{\bs{\alpha}_m}). \nonumber
\end{equation}
The above generalises the notion of residual entropy per individual variable
given in Section~\ref{sec:shared_randomness}~\footnote{In effect,
$R_{\bs{\alpha}_k}$ represents the portion of the entropy of the $k$-th cell
that is not shared with other cells}. With this, we introduce weights to each
edge of $\mathcal{G}_n$ based on residuals, given by
\begin{equation}\label{weight_r}
v_\text{r}(\pi_1,\pi_2) \coloneqq R(\pi_1)  - R(\pi_2)
\enspace.
\end{equation}
As residual entropy decreases when the partition is refined (see
Appendix~\ref{app:residual}), in this case one can illustrate the corresponding
DAG by placing nodes with more cells in lower positions (see lower half of
Figure~\ref{fig:diagramD3}).

Conveniently, for every edge $v_h$ and $v_r$ correspond to a mutual information
or a conditional mutual information term, respectively. This is illustrated in
the edges of Figure \ref{fig:diagramD3} and formalised in the Appendix.

The next result shows that the weights $v_\text{h}$ and $v_\text{r}$ provide
decompositions for $C(\bs{X}^n)$ and $B(\bs{X}^n)$, respectively.

\begin{lemma}\label{prop:TC_DCT}
Every path $\mathtt{p} \in \texttt{P}(\pi_\text{source},\pi_\text{sink})$
provides the following decompositions:
\begin{align}
W(\mathtt{p};v_\text{h}) &= C(\bs{X}^n) \nonumber\\
W(\mathtt{p};v_\text{r}) &= B(\bs{X}^n) \nonumber~ .
\end{align}
\end{lemma}
\begin{proof}
See Appendix~\ref{app:TC_DCT}.
\end{proof}

\begin{example}\label{ex:TC_dec}
For the case of $n=3$, there are three paths joining source and sink:
\begin{align*}
\texttt{p}_1=&\{ \texttt{(123), (1|23), (1|2|3)} \}, \\
\texttt{p}_2=&\{ \texttt{(123), (2|13), (1|2|3)} \}, \\
\texttt{p}_3=&\{ \texttt{(123), (3|12), (1|2|3)} \}.
\end{align*} 
Lemma~\ref{prop:TC_DCT} shows that
$C(\bs{X}^3)=W(\texttt{p}_i;v_\text{h})$ and
$B(\bs{X}^3)=W(\texttt{p}_i;v_\text{r})$ for $i\in\{1,2,3\}$, which provides
the following decompositions:
\begin{align*}
C(\bs{X}^3) & =I(X_i;X_j,X_k) + I(X_j;X_k) ~ , \\
B(\bs{X}^3) &= I(X_i;X_j,X_k) + I(X_j;X_k|X_i) ~ .
\end{align*}
\end{example}

\subsubsection{Lattice decomposition of $\Omega(\bs{X}^n)$}
\label{sec:omega_lattice}

Let us now leverage the decompositions presented in the previous subsection to
develop decompositions for the O-information. For this, let us first introduce
a new assignment of weights for the edges of $\mathcal{G}_n$, given by
\begin{equation}
v_\text{s}(\pi_1,\pi_2) := v_\text{h}(\pi_1,\pi_2) - v_\text{r}(\pi_1,\pi_2)
\enspace.
\end{equation}
In contrast with Eqs.~\eqref{weight_h} and \eqref{weight_r}, these weights can
attain negative values. The following key result shows that the weights $v_s$
provide a decomposition of $\Omega(\bm X^n)$.

\begin{proposition}\label{teo:1}
Every path $\mathtt{p} \in \texttt{P}(\pi_\text{source},\pi_\text{sink})$
provides the following decomposition:
\begin{equation}\label{eq:sum_interacitoninfo}
W(\mathtt{p};v_\text{s}) = \Omega(\bs{X}^n)~. 
\end{equation}
Moreover, Eq.~\eqref{eq:sum_interacitoninfo} is a sum of interaction
information terms of the form in Eq.~\eqref{eq:red-syn}.
\end{proposition}

\begin{proof}
See Appendix~\ref{app:theo:1}.
\end{proof}

This finding extends property (C) of Lemma~\ref{prop:basic} by showing that the
O-information can always be expressed as a sum of interaction information terms
of three sets of variables (see Corollary~\ref{cor:ass} below for an explicit
example of this). As a consequence, the O-information inherits the capabilities
of the triple interaction information for reflecting the balance between
synergies and redundancies, and is applicable to systems of any size.

An inconventient feature of partition lattices is that they grow
super-exponentially with system size~\footnote{The number of the nodes of
$\mathcal{G}_n$ grows with the \textit{Bell numbers}, known for their
super-exponential growth rate~\cite{comtet2012advanced}. To find the number of
paths in $\mathtt{P}(\pi_\text{source},\pi_\text{sink})$, note that if one
starts from the sink and moves towards the source, every step corresponds to
merging two cells into one. Therefore, as selecting two out of $m$ cells gives
${m\choose 2}$ choices, the total number of paths is given by
\unexpanded{\begin{equation*} |\texttt{P}(\pi_\text{source}, \pi_\text{sink})|
= \sum_{m=2}^n {m \choose 2} = \frac{n!(n-1)!}{2^{n-1}} \enspace,
\end{equation*}}
which grows faster than the Bell numbers.}, and hence heuristic methods for
exploring them are necessary. A particularly interesting sub-family of
$\texttt{P}(\pi_\text{source},\pi_\text{sink})$ are the ``assembly paths,''
which have the form (up to re-labelling)
\begin{equation}\label{eq:ass}
\texttt{p}_\text{a} = \{ (12\dots n), (12\dots (n-1)|n), \dots, (1|2|\dots|n)\}.
\end{equation}
These paths can be thought of as the process of first separating $X_n$ from the
rest of the system, then $X_{n-1}$, and so on. Conversely, by considering them
backwards, one can think of these paths as first connecting $X_1$ and $X_2$,
then connecting $X_3$ to $\bm X^2$, and so on -- i.e. as assembling the system
by sequentially placing its pieces together. The following corollary of
Proposition~\ref{teo:1} presents useful decompositions of $C(\bs{X}^n)$,
$B(\bs{X}^n)$, and $\Omega(\bs{X}^n)$ in terms of assembly paths.
\begin{corollary}\label{cor:ass}
  For an assembly path as given in Eq.~\eqref{eq:ass}, the corresponding
decompositions of the total correlation, binding entropy and O-information are
\begin{align}
C(\bs{X}^n) &= \sum_{i=2}^n I(X_i;\bs{X}^{i-1})\enspace,\label{bound_TC} \\
B(\bs{X}^n) &= I(X_n;\bs{X}^{n-1}) + \sum_{j=2}^{n-1} I(X_j;\bs{X}^{j-1}|\bs{X}_{j+1}^n), \label{bound_DTC}\\
\Omega(\bs{X}^n) &= \sum_{k=2}^{n-1} I(X_k;\bs{X}^{k-1} ; \bs{X}_{k+1}^n) \label{bound_o-info}
\enspace,
\end{align}
with $\bs{X}^n_{k} = (X_k,X_{k+1},\dots,X_n)$ and $\bs{X}^k = (X_1,\dots,X_k)$.
\end{corollary}

As a concluding remark, let us note that the decompositions presented by
Corollary~\ref{cor:ass} are valid for any relabeling of the indices (i.e. any
ordering of the system's variables). This property is a direct consequence of
the lattice construction developed in this subsection, which plays an important
role in the following sections.

\section{Understanding the O-information}
\label{sec:4}

By definition, $\Omega >0$ implies that the interdependencies are better
described as shared randomness, while $\Omega <0$ implies that they are better
explained as collective constraints. In this section we explore this further,
examining what the magnitude of $\Omega$ tells us about the system.

Through this section we use the shorthand notation $|\mathcal{X}| \coloneqq
\max_{j=1,\dots,n} |\mathcal{X}_j|$ for the cardinality of the largest alphabet
in $\bm X^n$.

\subsection{Characterising extreme values of $\Omega$}
\label{sec:omega_extremes}

Let us explore the range of values that the O-information can
attain. As a first step, Lemma~\ref{prop:bounds} provides bounds
for $C(\bm X^n)$, $B(\bm X^n)$, and $\Omega(\bm X^n)$. 
\begin{lemma}\label{prop:bounds}
The following bounds hold:
\begin{itemize}
  \item $(n-1) \log |\mathcal{X}| \geq C(\bs{X}^n) \geq 0$,
  \item $(n-1) \log |\mathcal{X}| \geq B(\bs{X}^n) \geq 0$,
  \item $n \log | \mathcal{X} | \geq C(\bm X^n) + B(\bm X^n) \geq 0$,
  \item $(n-2) \log |\mathcal{X}| \geq \Omega(\bs{X}^n) \geq (2-n) \log |\mathcal{X}|$.
\end{itemize}
Moreover, these bounds are tight.
\end{lemma}
\begin{proof}
See Appendix~\ref{app:bounds}.
\end{proof}

Let us introduce some nomenclature. A random binary vector $\bs{X}^n$ is said
to be a ``$n$-bit copy'' if $X_1$ is a Bernoulli random variable with parameter
$p=1/2$ (i.e. a \emph{fair coin}) and $X_1=X_2=\dots=X_n$. Also, a random
binary vector $\bs{X}^n$ is said to be a ``$n$-bit \texttt{xor}'' if
$\bs{X}^{n-1}$ are i.i.d. fair coins and $X_n=\sum_{j=1}^{n-1} X_j \pmod{2}$.
Our next result shows that these two distributions attain the upper and lower
bounds of the O-information.

\begin{proposition}\label{prop:iff}
Let $\bs{X}^n$ be a binary vector with $n\geq 3$. Then,
\begin{enumerate}
\item $\Omega(\bs{X}^n) = n-2$, if and only if $\bs{X}^n$ is a $n$-bit copy.
\item $\Omega(\bs{X}^n) = 2-n$, if and only if $\bs{X}^n$ is a $n$-bit \texttt{xor}.
\end{enumerate}
\end{proposition}

\begin{proof}
See Appendix~\ref{app:iff}. 
\end{proof}

\begin{corollary}
The same proof can be used to confirm that for variables with
$|\mathcal{X}_1|=\dots = |\mathcal{X}_n| =m$, the maximum $\Omega(\bs{X}^n) =
(n-2)\log m$ is attained by variables which are a copy of each other, while the
minimum $\Omega(\bs{X}^n) = (2-n)\log m$ corresponds to when $\bs{X}^{n-1}$ are
independent and uniformly distributed and $X_n=\sum_{j=1}^{n-1} X_j \pmod{m}$.
\end{corollary}

Proposition~\ref{prop:iff} points out an important difference betwen the
O-information and the interaction information: if $\bm X^n$ is an $n$-bit
\texttt{xor} then $\Omega(\bm X^n) = 2-n$ is consistently negative and
decreasing with $n$, while $I(X_1;\dots;X_n) = (-1)^{n+1}$ oddly oscillates
between $-1$ and $+1$. This result also points out the convenience of merging
$C(\bm X^n)$ and $B(\bm X^n)$ into $\Omega(\bm X^n)$, as only the latter has
the $n$-bit copy and the $n$-bit \texttt{xor} as unique extremes.

Finally, note that $\Omega$ is continuous over small changes in $p_{\bs{X}^n}$,
as it can be expressed as a linear combination of Shannon entropies (see
Definition~\ref{def:omega}). Therefore, Proposition~\ref{prop:iff} guarantees
that distributions that are similar to a $n$-bit copy have a positive
O-information, while distributions close to a $n$-bit \texttt{xor} have
negative O-information.

\subsection{Statistical structures across scales}\label{sec:scales}

In this section we study how the O-information is related to statistical
structures of subsets of $\bm X^n$ -- i.e. structures at different scales of
the system. For simplicity, we assume in this subsection that $|\mathcal{X}|$
is finite.

In the next proposition we present some fundamental restrictions between the
total correlation of subsystems and the value of $\Omega(\bm X^n)$.

\begin{proposition}\label{theo:bounds}
If $\Omega(\bs{X}^n)\geq 0$, then for all $m\in [n-1]$
\begin{equation}
\min_{|\bs{\gamma}|=m}C(\bs{X}^{\bs{\gamma}})
\geq
\Omega(\bs{X}^n) - (n-m-1)\log|\mathcal{X}| ~ .
\label{eq:11}
\end{equation}
If $\Omega(\bs{X}^n)\leq 0$, then for all $m\in [n-1]$
\begin{equation}
\max_{|\bs{\gamma}|=m}C(\bs{X}^{\bs{\gamma}})
\leq
\Omega(\bs{X}^n) +(n-2)\log|\mathcal{X}| ~ .
\label{eq:12}
\end{equation}
Both bounds are tight if $|\Omega| \geq (n-m+1) \log|\mathcal{X}|$.
\end{proposition}
\begin{proof}
See Appendix~\ref{app:bounds}.
\end{proof}

\begin{corollary}\label{cor:x}
The following bounds hold for all $\bs{\gamma}\in\{1,\dots,n\}$ with $|\bs{\gamma}|=m$:
\begin{align}
\min&\left\{m-1, \frac{\Omega(\bs{X}^n)}{\log|\mathcal{X}|} + (n-2) \right\} 
\geq \frac{C(\bs{X}^{\bs{\gamma}})}{\log|\mathcal{X}|} \nonumber\\
&\geq \max\left\{ 0, \frac{\Omega(\bs{X}^n)}{\log|\mathcal{X}|} - (n-m-1) \right\}. \nonumber
\end{align}
\end{corollary}

Corollary~\ref{cor:x} shows that positive values of $\Omega$ constrain
subgroups to be correlated: if $\Omega(\bs{X}^n) \geq (n-m-1)\log|\mathcal{X}|$
then all groups of $m$ or more variables must have some statistical dependency.
Negative values of $\Omega$, on the other hand, impose limits on the allowed
correlation strength: if $\Omega(\bs{X}^n) \leq -(n-m-1)\log|\mathcal{X}|$ then
the correlation of all groups of $m$ or more variables is upper-bounded. As an
example, for $|\mathcal{X}|=2$ and $m=2$ the bounds given in
Corollary~\ref{cor:x} are
\begin{align*}
\max&\left\{1, \Omega(\bs{X}^n) + n-2 \right\} 
\geq I(X_i;X_j) \\
&\geq \min\left\{ 0, \Omega(\bs{X}^n) - (n-3) \right\} ~,
\end{align*}
for all $i,j\in\{1,\dots,n\}$, which shows that the bounds related to $\Omega$
are only active when $n-3 \leq |\Omega| \leq n-2$.

In conclusion, the sign of $\Omega$ determines whether the constraint is a
lower or upper bound, and $|\Omega|$ determines which scales of the system are
affected, with smaller groups being harder to constrain -- i.e. requiring
higher absolute values of $\Omega$. The relationship between the system's
scales and the values of $\Omega$ is illustrated in Figure~\ref{fig:bounds1}.

\begin{figure}[ht]
  \centering
  \includegraphics{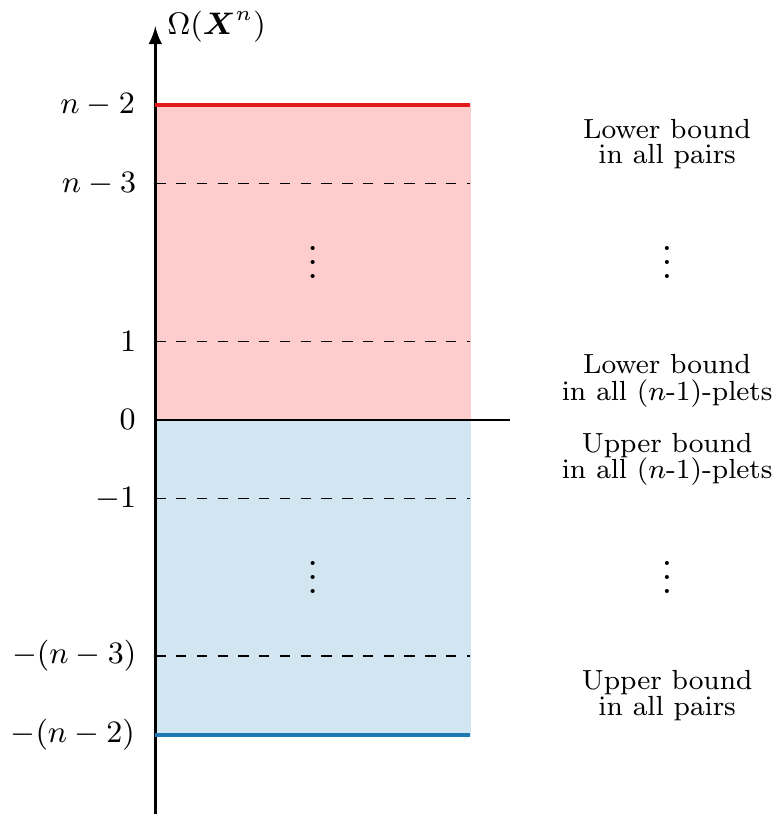}

  \caption{Values of the O-information impose limits on the
    strength of interactions -- as measured by $C(\bs{X}^{\bs{\gamma}})$ -- at
    different scales.  Positive (negative) values of $\Omega$ put lower (upper)
    bounds on subsets of $\bm X^n$, and higher absolute values of $\Omega$ put
  bounds on subsystems of smaller sizes.}

\label{fig:bounds1}
\end{figure}

The next result corresponds to the converse of Corollary~\ref{cor:x}, and shows
how interactions at different scales limit the achievable values of $\Omega$.

\begin{corollary}\label{cor:y}
For a given $\bs{\gamma}\subset\{1,\dots,n\}$ with $|\bs{\gamma}|=m$, the
following bounds on $\Omega$ hold:
\begin{align}
n - m -1   + \frac{ C(\bs{X}^{\bs{\gamma}}) }{ \log |\mathcal{X}| } 
\geq 
\frac{ \Omega(\bs{X}^n) }{ \log |\mathcal{X} |} 
\geq -(n-2) + \frac{ C(\bs{X}^{\bs{\gamma}}) }{ \log |\mathcal{X}|}. \nonumber
\end{align}
\end{corollary}

By comparing it with Lemma~\ref{prop:bounds}, this result shows that a large
$C(\bs{X}^{\bs{\gamma}})$ does not allow $\Omega$ to reach its lower bound. On
the other hand, small values of $C(\bs{X}^{\bs{\gamma}})$ decrease the upper
bound, forbidding high values of $\Omega$. Additionally, note that fixing
the value of only one subset of $m$ variables reduces the range of values of
$\Omega$ from $2(n-2)$ to $2(n-2) - (m-1)$. The following example illustates
these findings.

\begin{example}
Let us consider a system $\bs{X}^n$ of binary variables, two of which are
related by the marginal distribution
\begin{equation*}
p_{X_1 X_2}(x_1, x_2) = \frac{(1-\eta)^{1 - |x_1-x_2|} \eta^{|x_1-x_2|} } {2} ~ .
\end{equation*}
That is, $X_1$ and $X_2$ are fair coins linked by a binary symmetric channel
with crossover probability $\eta$~\cite[Sec. 7]{cover2012elements}. Hence,
$C(\bs{X}^2) = I(X_1;X_2) = 1 - H(\eta)$, with $H(\eta)= -\eta \log \eta -
(1-\eta) \log (1-\eta)$ being the binary entropy function. By considering
$m=2$, Corollary~\ref{cor:y} states that
\begin{equation}
n-2 - H(\eta) \geq \Omega(\bs{X}^n) \geq -\big (n-3+H(\eta)\big) ~ , \nonumber
\end{equation}
which is illustrated in Figure~\ref{OmegaBSC}. Moreover, using
Eq.~\eqref{bound_o-info} one can verify that the upper bound (solid red line)
is attained when $X_2=X_3=\dots=X_n$, while the lower bound (solid blue line)
is attained when $X_3,\dots,X_{n-1}$ are independent fair coins and
$X_n=\sum_{j=1}^{n-1} X_j \pmod{2}$ \footnote{Interestingly, despite the
correlation between $X_1$ and $X_2$, an $n$-bit \texttt{xor} still enables the
most synergistic configuration attainable.}.
\begin{figure}[ht]
  \centering
   \includegraphics{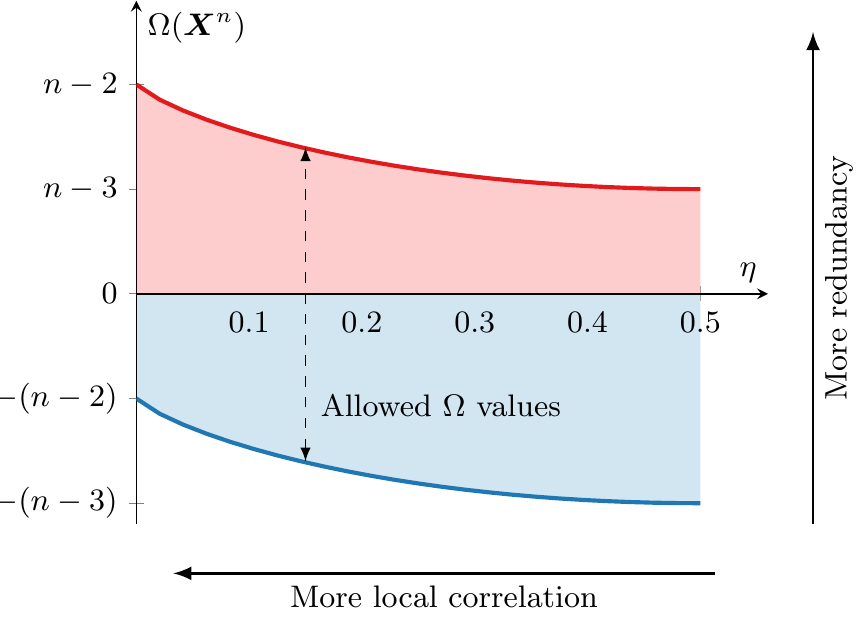}
  \caption{Bounds of the O-information when two variables are connected via a
  binary symmetric channel with crossover probability $\eta$.}
  \label{OmegaBSC}
\end{figure}
\end{example}

\subsection{$\Omega$ as a superposition of tendencies}

This subsection explores sufficient conditions that make a system have a small
O-information. As a preliminary step, the next result shows that $\Omega$ is
additive for systems with independent subsystems.

\begin{lemma}\label{prop:indepo}
If $p_{\bs{X}^n}(\bs{x}^n) = \prod_{k=1}^m p_{\bs{X}^{\bs{\alpha}_k}} (\bs{x}^{\bs{\alpha}_k})$ for some partition  $\pi=(\bs{\alpha}_1|\dots|\bs{\alpha}_m)$, then 
\begin{equation}
\Omega(\bs{X}^n) = \sum_{k=1}^m \Omega(\bs{X}^{\bs{\alpha}_k})~.\nonumber
\end{equation}
\end{lemma}
\begin{proof}
Let us consider the case $\pi=(\bs{\alpha}_1, \bs{\alpha}_2)$, as the
general case is then guaranteed by induction. Using Eqs.~\eqref{bound_TC} and
\eqref{bound_DTC} it is direct to check that, due to the independence,
$C(\bs{X}^n) = C(\bs{X}^{\bs{\alpha}_1}) + C(\bs{X}^{\bs{\alpha}_2})$ and
$B(\bs{X}^n) = B(\bs{X}^{\bs{\alpha}_1}) + B(\bs{X}^{\bs{\alpha}_2})$. Then,
the desired result follows from the fact that $\Omega(\bm X^n)=C(\bm X^n) - B(\bm
X^n)$.
\end{proof}

\begin{corollary}\label{cor:pairwise}
$\Omega(\bs{X}^n) = 0$ for all systems whose joint distribution can be
factorised as
\begin{equation}
p_{\bs{X}^n}(\bs{x}^n) = \prod_{k=1}^{n/2} p_{X_{2k-1} X_{2k}} (x_{2k-1},x_{2k})~.\label{eq:qerty}%
\end{equation}
\end{corollary}
\begin{proof}
Using Eq.~\eqref{eq:qerty} and Lemma~\ref{prop:indepo} we find that
\begin{equation}
\Omega(\bs{X}^n) =  \sum_{k=1}^{n/2} \Omega(X_{2k-1}, X_{2k}) = 0
\enspace,\nonumber
\end{equation}
where the last equality is a consequence of the O-information being zero for
sets of two variables, as shown in Proposition~\ref{prop:basic}.
\end{proof}

Corollary~\ref{cor:pairwise} states that having disjoint pairwise interactions
is a sufficient condition for $\Omega=0$ to hold. However, this condition is
not necessary: from Lemma~\ref{prop:indepo} we can see that a system composed
by redundant ($\Omega>0$) and synergistic ($\Omega<0$) subsystems can attain
zero net O-information due to ``destructive interference.''

As a consequence, the O-information can be understood as the result of a
superposition of behaviours of subsystems. Therefore, $\Omega=0$ can take place
in two qualitatively different scenarios: systems in which redundancies and
synergies are balanced, or systems with only disjoint pairwise effects. Some of
these cases can be resolved by considering the information diagram of $C(\bm
X^n)$ and $B(\bm X^n)$ (c.f. Figure~\ref{fig:diagramD3}), or by studying the
O-information of parts of the system. However, it is important to remark that
redudancy and synergy can coexist either in disjoint subsystems or within the
same variables. An insightful example of the latter case can be found in
Ref.~\cite[Section 2]{james2017multivariate}.

As a final remark, note that systems where pairwise interdependencies are overlapping
(e.g. pairwise maximum entropy models~\cite{cofre2018}) cannot be factorised as
required by Corollary~\ref{cor:pairwise}, and hence can have either positive or
negative O-information~\footnote{For a detailed discussion of this issue for
the case of three variables see~\cite[Sec. 5]{rosas2016understanding}.}.

\section{Relationship with other notions of high-order effects}
\label{sec:other_metrics}

\subsection{High-order interactions in statistical mechanics}

A popular approach to address high-order interactions in the statistical
physics literature is via Hamiltonians that include interaction terms with
three or more variables~\cite{schneidman2003network}. For example, systems of
$n$ spins (i.e. $\mathcal{X}_i = \{-1,1\}$ for $i=1,\dots,n$) that exhibit
$k$-th order interactions are usually represented by probability distributions
of the form
\begin{equation}\label{eq:gibbs}
p_{\bm X^n}(\bm x^n) = \frac{ e^{-\beta \mathcal{H}_k(\bm x^n)}}{Z}~,
\end{equation}
where $\beta$ is the inverse temperature, $Z$ is a normalization constant, and
$\mathcal{H}(\bm x^n)$ is a Hamiltonian given by
\begin{align}
  \mathcal{H}_k(\bm x^n) = &- \sum_{i=1}^n J_i x_i - \sum_{i=1}^{n-1} \sum_{j=i+1}^n J_{i,j} x_i x_j \nonumber \\
& \dots - \sum_{|\bs{\gamma}|=k} J_{\bs{\gamma}} \prod_{i\in \bs{\gamma}} x_i~, \nonumber
  \label{eq:hamiltonian}
\end{align}
with the last sum runing over all subsets $\bs{\gamma}\subseteq \{1,\dots,n\}$
of size $|\bs{\gamma}| = k$. According to Eq.~\eqref{eq:gibbs}, configurations
with lower $\mathcal{H}_k(\bm x^n)$ are more likely to be visited. Note that
$J_i$ quantify external influences acting over individual spins, while
$J_{\bs{\gamma}}$ for $|\bs{\gamma} |\geq 2$ represent the strength of the
interactions; in particular, if $J_{i,k} > 0 $ then the pair $X_i,X_k$ tend to
be aligned, while if $J_{i,k} < 0$ they tend to be anti-aligned. As a matter of
fact, $\bm X^n$ are independent if and only if $J_{\bs{\gamma}} = 0$ for all
$\bs{\gamma}$ with $|\bs{\gamma}|\geq 2$. Models with $k$-th order interactions
have been studied via the maximum entropy
principle~\cite{schneidman2003network}, information
geometry~\cite{amari2001information} and PID~\cite{olbrich2015information}.

Considering the results presented in previous sections, one could expect that
systems with high-order interactions (i.e. large $k$) should attain lower
values of $\Omega$ than systems with low-order interactions (i.e. small $k$).
To confirm this hypothesis, we studied ensembles of systems with $k$-th order
interactions, and analised how the value of $\Omega$ is influenced by $k$. For
this, we considered random Hamiltonians with $J_{\bs{\gamma}}$ drawn i.i.d.
from a standard normal distribution and $\beta=0.1$.

In agreement with intuition, results show that $\Omega$ is usually very close
to zero for $k=2$, and becomes negative as $k$ grows
(Figure~\ref{fig:hamiltonian}). These results suggest that the notion of
synergy measured by $\Omega$ is consistent with the traditional ideas of
high-order interactions from statistical physics.
\begin{figure}[ht]
  \centering
  \includegraphics{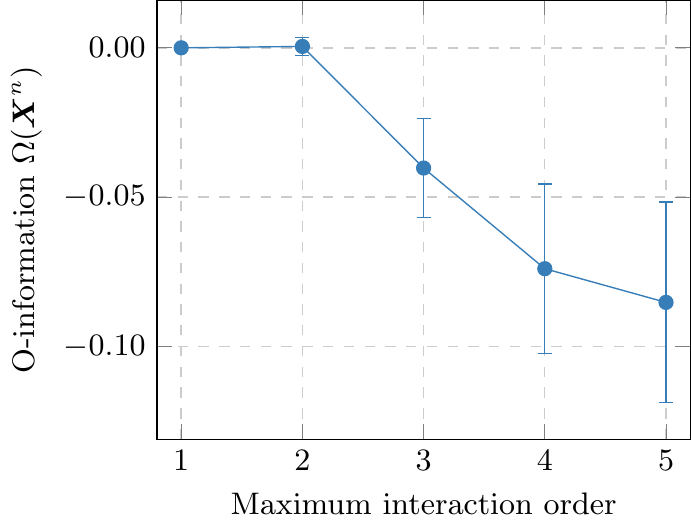}
  \caption{Mean value and confidence intervals of ensembles of systems of $n=5$ spins with randomly generated Hamiltonians. By including high-order interaction terms, net synergy increases and $\Omega$ decreases.}
  \label{fig:hamiltonian}
\end{figure}

\subsection{Complexity and integration}
\label{sec:TSE}

In their seminal 1994 article, Tononi, Edelman, and Sporns devised a measure of
complexity (henceforth called \textit{TSE complexity}) to describe the
interplay between local segregation and global
integration~\cite{tononi1994measure,tononi1998complexity}. The TSE complexity
is defined as
\begin{equation}
  \text{TSE}(\bm X^n) \coloneqq \sum_{k=1}^{n} \left[ \frac{k}{n} C(\bm X^n) - C_n(k) \right]~,
\label{eq:tse}
\end{equation}
where $C_n(k) = {n\choose k}^{-1}\sum_{|\bs{\gamma}| = k} C(\bm
X^{\bs{\gamma}})$ is the average total correlation of the subsets
$\bs{\gamma}\subseteq \{1,\dots,n\}$ of size $|\bs{\gamma}| = k$. By measuring
the convexity of $C_n(k)$, the TSE complexity attempts to distinguish scenarios
that exhibit ``relative statistical independence of small subsets of the system
[\dots] and significant deviations from independence of large subsets''
\cite[Abstract]{tononi1994measure}, in the same spirit as our motivation behind
$\Omega$ above.

To study the relationship between the TSE complexity and the O-information, it
is useful to consider an alternative expression of the former:
\begin{equation}\label{eq:TSE_MI}
  \text{TSE}(\bm X^n) = \sum_{k=1}^{\lfloor n/2 \rfloor} {n\choose k}^{-1} \sum_{|\bs{\gamma}| = k} I(\bm X^{\bs{\gamma}}; \bm X_{-\bs{\gamma}}^n) ~ ,
\end{equation}
\noindent where $X_{-\bs{\gamma}}^n$ represents all the variables that are not
in $\bs{\gamma}$, and $\lfloor \cdot \rfloor$ is the floor function. By noting
the similarities between Eq.~\eqref{eq:TSE_MI} and the sum of $B$ and $C$,
\begin{equation}\label{eq:sumCB}
C(\bm X^n) + B(\bm X^n) = \sum_{i=1} I(X_i; \bm X_{-i}^n) ~ , 
\end{equation}
together with the fact that $\text{TSE}(\bm X^3)=\frac{1}{3}\big[ C (\bm X^3) +
B(\bm X^3) \big]$, we can hypothesise that, qualitatively,
\begin{equation}\label{eq:approx}
\text{TSE}(\bm X^n) \propto C(\bm X^n) + B(\bm X^n) ~ .
\end{equation}
Monte Carlo simulations show that this approximation is justified: when
evaluated on distributions $p_{\bm X^n}$ sampled uniformly at random from the
probability simplex, the correlation of Eq.~\eqref{eq:approx} and TSE is
consistently above $0.97$ (Figure~\ref{fig:tse}). Moreover,
Eq.~\eqref{eq:approx} outperforms other proposed approximations of the TSE
complexity~\footnote{\unexpanded{In \cite[Fig.~2]{tononi1998complexity} the
binding entropy (under the name ``interaction complexity'') is proposed as a
metric ``related but not identical to neural complexity.'' Numerical
evaluations show that the combination of total correlation and binding entropy,
as proposed in~\eqref{eq:approx}, is a more accurate approximation for the TSE
complexity (results not shown).}}.

\begin{figure}[ht]
  \centering
  \includegraphics{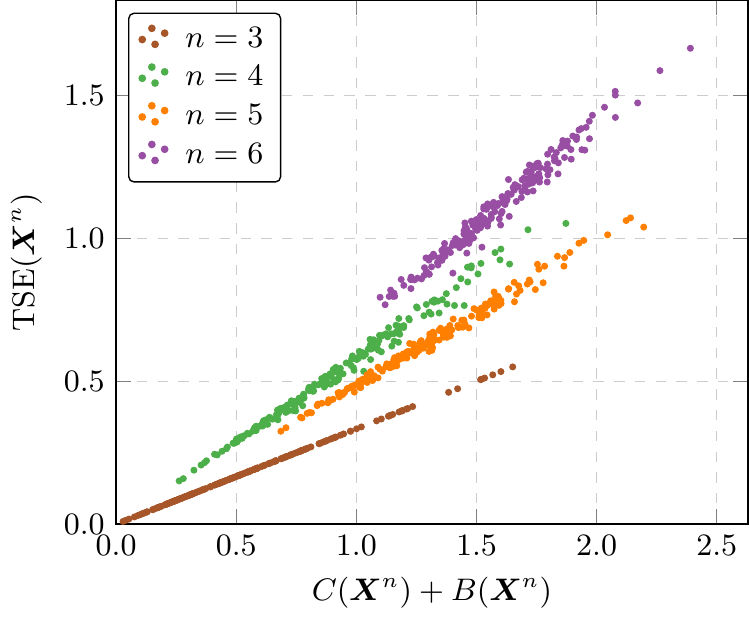}
  \caption{(Color) The sum of the total correlation and the binding entropy is a very accurate approximation of the TSE complexity. Each dot is a probability distribution over $n$ bits sampled uniformly at random from the probability simplex.}
  \label{fig:tse}
\end{figure}

Figure \ref{fig:tse} and Eq.~\eqref{eq:approx} suggest that the TSE complexity
is large when either the shared randomness or the collective constraints are
large. As a more direct example, we evaluate TSE in a distribution given by a
linear mixture of a 3-bit copy and a 3-bit \texttt{xor}, showing that TSE has
exactly the same value in both extremes, and hence that it conflates redundancy
with synergy (Figure~\ref{fig:conflate}).

Taken together, our results show that the TSE complexity is a good metric of
overall integration between parts of the system, but it generally fails to
detect synergistic phenomena. Overall, the fact that
\begin{equation}\label{eq:OvsTSE}
  \begin{split}
    \Omega &= C - B~, \\
    \text{TSE} &\propto C + B~,
  \end{split}
\end{equation}
suggests that the TSE complexity and the O-information are complementary,
corresponding to an insightful ``change of basis'' from an elementary
constraints vs randomness representation. Effectively, while both $C$ and $B$
provide two measures of roughly the same phenomenon (interdependency strength),
$\Omega$ and $\text{TSE}$ refer to different aspects: TSE gives an overarching
account of the strength of the interdependencies within $\bm X^n$, and $\Omega$
indicates whether these correlations are predominantly redundant or
synergistic.

\begin{figure}[ht]
  \centering
  \includegraphics{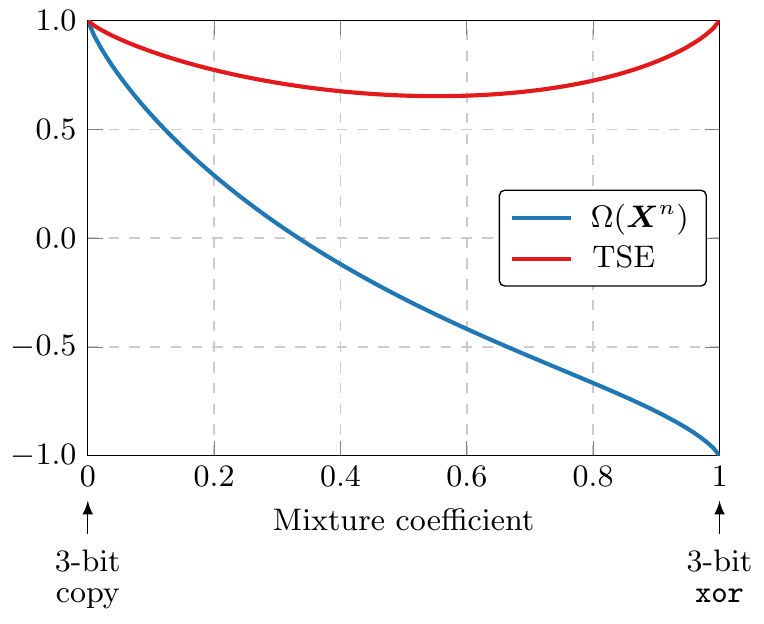}
  \caption{(Color) TSE and $\Omega$ evaluated on a distribution resulting from a linear
  mixture between a copy (left) and an \texttt{xor} (right), showing that the
TSE complexity conflates synergy and redundancy. Figure shows the case $n=3$,
but results are qualitatively similar for larger systems}
  \label{fig:conflate}
\end{figure}

\section{Case study: Baroque music scores}
\label{sec:5}

To illustrate the proposed framework in a data-driven application, this section
presents a study of the multivariate statistics of musical scores from the
Baroque period. In the sequel, Section~\ref{sec:method} describes the procedure
to obtain and analyse the data. Results are then presented in
Section~\ref{sec:results}. These results are a brief demonstration of the value
for the O-information for practical data analysis.

\subsection{Method description}
\label{sec:method}

\subsubsection{Data}

Our analysis focuses on two sets of repertoire: the well-known chorales for
four voices by Johann Sebastian Bach (1685-1750), and the Opus 1 and 3-6 by
Arcangelo Corelli (1653-1713). All of these works correspond to the Baroque
period (approx. 1600--1750), which is characterised by elaborate counterpoint
between melodic lines. Baroque music usually exhibits a balance in the interest
and richness of the parts of all the involved instruments, contrasting with the
subsequent Classic (1730--1820) and Romantic (1780--1910) periods where higher
voices tend to take the lead.

Our analysis is based on the electronic scores publicly available at
\url{http://kern.ccarh.org}. We focused on scores with four melodic lines: four
voices (soprano, alto, tenor and bass) in the case of Bach's chorales, and four
string instruments (1\textsuperscript{st} violin, 2\textsuperscript{nd} violin,
viola and cello) in the case of Corelli's pieces. The scores were pre-processed
in Python using the \texttt{Music21} package
(\url{http://web.mit.edu/music21}), which allowed us to select only the pieces
writen in Major mode and to transpose them to C Major. The melodic lines were
transformed into time series of 13 possible values (one for each note plus one
for the silence), using the smallest rhythmic duration as time unit. This
generated $\approx 4\times 10^4$ four-note chords for the chorales, and
$\approx 8\times 10^4$ for Corelli's pieces. With these data, the joint
distribution of the values for the four-note chords was estimated using their
empirical frequency~\footnote{\unexpanded{Regularisation methods (such as
Laplace smoothing) were found to have strong effects the results. We decided
not to use such methods, as some chords (e.g. C-C$\sharp$-D-D$\sharp$) are just
not going to take place in the Baroque repertoire.}}.

\subsubsection{Research questions and tools}

We focus on the multivariate statistics of the harmonic structures of these
pieces. In particular, we ask to what extent the notes played simultaneously by
different instruments are redundant or synergistic. Our study focuses
exclusively on harmony and chords, leaving melodic properties to future
studies.

Let us denote by $\bs{X}^4$ the random vector of notes, where
$|\mathcal{X}|=13$. We first compute the marginal entropy of each voice,
$H(X_k)$, which is an indicator of harmonic richness. We also compute the
O-information of the ensemble $\Omega(\bs{X}^4)$, which determines the dominant
behaviour. Interestingly, for $n=4$ the decomposition in
Eq.~\eqref{bound_o-info} yields
\begin{equation*}
\Omega(\bs{X}^4) = I(X_i;X_j;X_k,X_l) + I(X_k;X_l;X_i,X_j)
\end{equation*}
for $\{i,j,k,l\} = \{1,2,3,4\}$. One can gain a fine-grained view of $\Omega$
by considering these interaction information terms, which can be seen as local
contributions to $\Omega$. More formally, we define the \emph{local
O-information} between $X_i$ and $X_j$ as
\begin{align}
  \omega_{ij}(\bm X^n) \coloneqq I(X_i;X_j;\bm X^n_{-ij}) ~ ,
  \label{eq:local_o-info}
\end{align}
\noindent such that $\Omega$ can be decomposed as a sum of local $\omega$.
Interestingly, these local terms could be of the opposite sign to the global
$\Omega(\bs{X}^n)$, indicating local synergy (or redudancy) between two
components within a predominantly redundant (or synergistic) system.

Since all the $X_k$ take values among alphabets of cardinality $13$, we perform
all computations employing logarithms to base $13$, so that $H(X_k) \leq 1$. We
call this unit a \emph{mut}, for \emph{musical bit}.

\subsection{Results}
\label{sec:results}

\begin{figure*}[t!]
  \centering
  \includegraphics{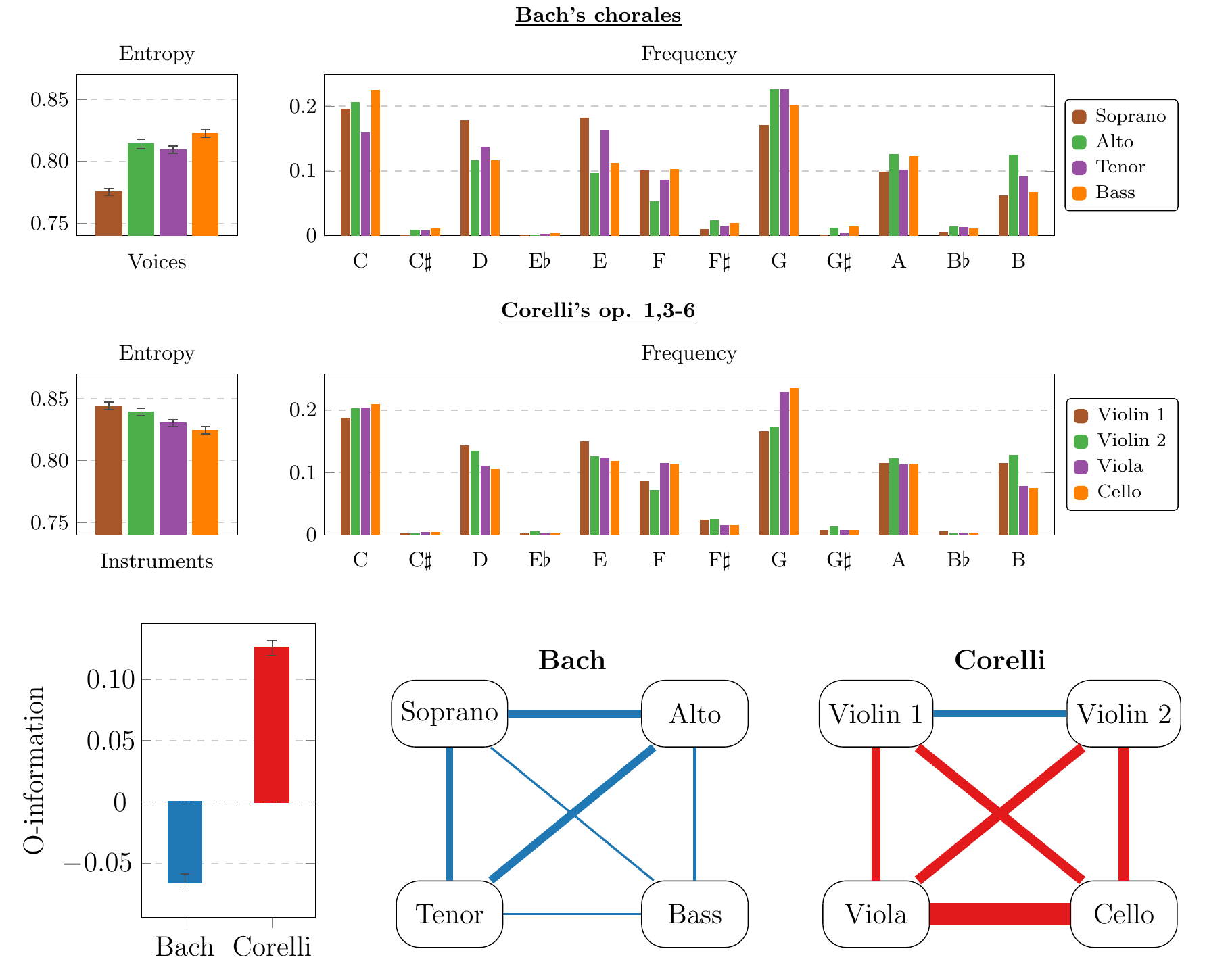}
  \caption{(Color) \textit{Above:} Entropy of the frequencies of appeareance of each
    note in the studied pieces of Bach and Corelli, measured in muts (logarithm
    to base 13); standard errors were estimated via circular block-bootstrap. While the higher voices in Corelli have higher entropy, Bach's
    soprano has a lower entropy than all other voices. \textit{Below:} Global
    O-information (left), and networks of local O-information (middle, right)
    with red reflecting redundancy ($w_{ij} >0$) and blue synergy ($w_{ij}  <
    0$). While Bach's chorales are synergy-dominated, the pieces of Corelli are
    strongly redundant (mainly due to the viola and cello). 
}
  \label{fig:music_results}
\end{figure*}
 
By studying the entropies of each voice, our results confirm that the four
voices in these Baroque scores tend to have similar harmonical richness
(Figure~\ref{fig:music_results}, top left). In fact, their values are similar
(although slightly lower) than $\log_{13}7 \approx 0.845$ muts, which
corresponds to a uniform distibution over the seven notes of a major scale
(notes without sharp or flat). Also, our results show that the entropies in the
music of Corelli are higher for instruments with higher register (i.e. the
violins). In contrast, in Bach's music the soprano has significantly less
entropy than the other voices. This can be explained by the fact that Bach's
pieces were made to be used in public religious services, with the soprano
conveying a melodic line that was intended to be sung by the attendees -- and
hence its structure is simpler to make it easy to sing.

Most strikingly, our analyses of the multivariate structure of the pieces show
that Bach's chorales have negative O-information, suggesting that the harmonic
structure of these pieces is dominated by synergistic effects
(Figure~\ref{fig:music_results}, bottom left). This result is further confirmed
by the fact that all the local O-information terms are negative, which means
that the pairwise dependence between any pair of voices is comparatively
smaller than the global dependencies that exists within the group (see
Table~\ref{tab:sym}).

\begin{table}[htbp]
\begin{center}

  \caption{Multivariate statistics of Baroque repertoire. For each pair of
    voices or instruments, we report the mutual information (MI), conditional
    mutual information (CMI), and local O-information ($\omega_{ij}$).
    Quantities are measured in musical bits, or \emph{muts} (logarithm to base
    13).  Standard errors were estimated via circular block-bootstrap, and in all cases are
    below the least significant figure shown in the table.
}

\label{tab:sym}
\begin{tabular}{r r >{\hspace{5pt}}c >{\hspace{5pt}}c >{\hspace{5pt}}c}
\multicolumn{5}{c}{}\\
\multicolumn{5}{c}{\textbf{\underline{Bach's chorales}}}\\
\multicolumn{5}{c}{}\\
\specialrule{.13em}{.0em}{.15em}
& & MI & CMI & $\omega_{ij}$ \\
\specialrule{.05em}{.08em}{.07em}
\textbf{Soprano} & \textbf{Alto}  & ~ 0.14 ~ & ~ 0.19 ~ & -0.05 \\
\textbf{Soprano} & \textbf{Tenor} & ~ 0.12 ~ & ~ 0.16 ~ & -0.04 \\
\textbf{Soprano} & \textbf{Bass}  & ~ 0.15 ~ & ~ 0.16 ~ & -0.02 \\
\textbf{Alto}    & \textbf{Tenor} & ~ 0.17 ~ & ~ 0.22 ~ & -0.05 \\
\textbf{Alto}    & \textbf{Bass}  & ~ 0.15 ~ & ~ 0.17 ~ & -0.02 \\
\textbf{Tenor}   & \textbf{Bass}  & ~ 0.15 ~ & ~ 0.17 ~ & -0.02 \\
\specialrule{.13em}{.15em}{.0em} 
\multicolumn{5}{c}{}\\
\multicolumn{5}{c}{}\\
\multicolumn{5}{c}{\textbf{\underline{Corelli's op. 1,3-6}}}\\
\multicolumn{5}{c}{}\\
\specialrule{.13em}{.0em}{.15em}
& & MI & CMI & $\omega_{ij}$ \\
\specialrule{.05em}{.08em}{.07em}
\textbf{Violin\,1} & \textbf{Violin\,2}  & 0.071  & 0.115  & -0.04\hphantom{-} \\
\textbf{Violin\,1} & \textbf{Viola}    & 0.086  & 0.028  & 0.06  \\
\textbf{Violin\,1} & \textbf{Cello}    & 0.095  & 0.034  & 0.06  \\
\textbf{Violin\,2} & \textbf{Viola}    & 0.118  & 0.054  & 0.07  \\
\textbf{Violin\,2} & \textbf{Cello}    & 0.107  & 0.039  & 0.07  \\
\textbf{Viola}   & \textbf{Cello}    & 0.630  & 0.460  & 0.17  \\
\specialrule{.13em}{.15em}{.0em} 
\end{tabular}
\end{center}
\end{table}

In contrast, Corelli's pieces have positive O-information, suggesting that they
are dominated by a redundant component. Interestingly, the local O-information
has a positive value for all pairs except for violins 1 and 2. The strongest
O-information is the one between viola and cello, indicating that the parts of
these two instruments are highly redundant.

The redundancy in the pieces of Corelli can be explained by compositional
practices for intrumental music in the Baroque period. In fact, the original
score of many of the studied pieces was written for only three parts: two
solists and a bass line called ``basso continuo.'' This bass line was suposed
to be interpreted in different ways by the bass instruments, which in this case
correspond to viola and cello. Therefore, it is fair to say that these
instruments are redundant, as both of them are carrying the same bass line.
Despite this redundancy, the relationship between the violins is still
synergistic, which is appropiately captured by the negative value of their
local O-information.

The dominance of synergy in the case of Bach can be argued to serve an artistic
purpose -- in effect, in the Baroque period the aim was that each voice should
introduce unique elements into the piece. This goal could be achieved by
superposing unrelated melodies; however, the overall result might not have been
aesthetically pleasing due to the lack of global coordination. A synergistic
structure serves this purpose well, as it provides global constraints that
ensure collective coherence while imposing weak pairwise constraints.

\section{Conclusion}
\label{sec:discussion}

We introduced $\Omega(\bm X^n)$ as the difference between strength of the
collective constraints and the shared randomness in a multivariate system $\bm
X^n$. We argued that $\Omega$ captures the net balance between statistical
synergy and redundancy, since (i) it is a sum of triple interaction
informations, (ii) it is maximised (minimised) by an $n$-bit copy
(\texttt{xor}), and (iii) it imposes bounds over the intedependency allowed at
different scales. According to this framework, synergistic systems are
characterised by a large amount of shared randomness regulated by weak
collective constraints, which is consistent with recent approaches to study
emergence based on constructive logic~\cite{pascual2018constructive}. Moreover,
in deriving $\Omega$, we also provided a joint source of explanation for three
long-standing extensions of Shannon's mutual information (total correlation,
binding entropy, and interaction information) in terms of shared randomness and
collective constraints. The proposed framework is straightforward to generalise
to continuous variables and apply to neural data, which will be presented in a
separate publication.

The O-information was compared to other notions of high-order effects, most
notably the TSE complexity~\cite{tononi1994measure}. We found that TSE does not
measure statistical synergy as such, but total correlation strength. Moreover,
our analysis suggest that $\Omega$ and TSE are complementary metrics: TSE gives
an overarching account of the strength of the interdependencies within $\bm
X^n$, and the O-information reveals whether these correlations are
predominantly redundant or synergistic. We take this as a step towards a
multi-dimensional framework that allows for a finer and more subtle taxonomy of
complex systems.

Finally, we applied our framework to Baroque music scores and found
that Bach's chorales, unlike pieces by some of his contemporaries, are strongly
synergistic as measured by $\Omega$. Informally, we can speculate about the
artistic role of synergy: synergistic music (like Bach's) allows each voice to
contribute unique material while ensuring an overall harmonious integration of
the ensemble. This delicate balance has an intriguing similarity with the
coexistence of integration and differentiation in brain
activity~\cite{tononi1998complexity,balduzzi2008integrated}, suggesting 
unexplored relationships between music structure and neural organisation.

\section*{Acknowledgements}

The authors thank Shamil Chandaria, Alberto Pascual and Nicolas Rivera for
insightful discussions. Fernando Rosas was supported by the European Union's
H2020 research and innovation programme, under the Marie Sk\l{}odowska-Curie
grant agreement No. 702981.

\appendix

\section{Compatibility between $\Omega$ and prior work}
\label{app:selforg}

In prior work~\cite{rosas2018selforg}, we introduced $\psi(k)$ as
\begin{align}
  \psi(k) \coloneqq \max_{j\in\{1,\dots,n\}} \max_{\substack{\boldsymbol{\gamma} \subseteq \{1,\dots,n\}\\ |\bs{\gamma}|=k, j \notin \bs{\gamma}}} I(X_j ; \boldsymbol{X}^{\boldsymbol{\gamma}}) ~. \nonumber
\end{align}
The growth profile of this non-decreasing function was taken as an indicator of
the leading quality of the interdependency structure of $\bm X^n$, being
convexity associated with statistical synergy, and concavity with
redundancy~~\cite[Definition 2]{rosas2018selforg}.

The relationship between these ideas and the ones developed in this article can
be established by noting that convexity in $\psi(k)$ implies that small scales
of the system are relatively independent while large scales show correlation,
which -- due to the results of Section~\ref{sec:scales} -- is the key
characteristic of synergy-dominated systems. Conversely, concavity in $\psi(k)$
implies that some small groups of variables are highly correlated, which
implies a relatively high value of $C(\bm X^n)$ and $\Omega(\bm X^n)$.

To enable a quantitative comparison between $\psi(k)$ and $\Omega$, one can
quantify the convexity/concavity of the former by measuring the distance from
$\psi(k)$ to a straight line joining $\psi(1)$ and $\psi(n)$ as
\begin{equation}
\Psi(\bm X^n) \coloneqq \sum_{k=1}^n \left[ \psi(k) - \left( \frac{k}{n} \big[\psi(n)-\psi(1)\big] + \psi(1) \right) \right]~. \nonumber
\end{equation}
We computed $\Omega$ and $\Psi$ of binary systems of different sizes generated
randomly from a uniform distribution over the corresponding probability
simplex. Our results show a good agreement between these two metrics, which
confirms the analytic reasoning presented above.
\begin{figure}[ht]
  \centering
  \includegraphics{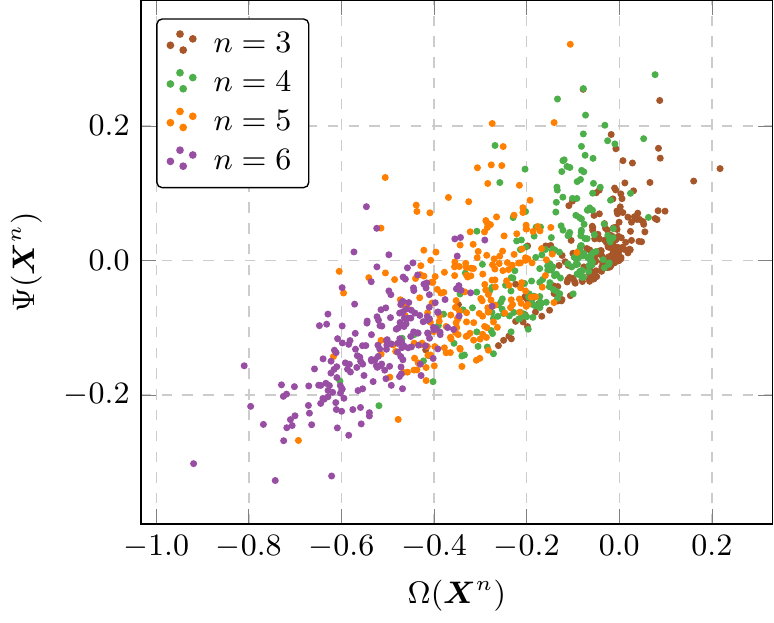}
  \caption{The O-information and $\Psi$ -- introduced in our previous work~\cite{rosas2018selforg} -- have good agreement.}
  \label{fig:selforg}
\end{figure}

In summary, $\Omega$ can be regarded as a formalisation of the intuitive
notions introduced in \cite{rosas2018selforg}. Moreover, $\Omega$ possesses
more theoretical properties than $\Psi$ and requires the calculation of a
smaller number of terms.

\section{$R(\pi)$ decreases for finer partitions}
\label{app:residual}

\begin{lemma}\label{lemmaX}
Let us consider two partitions $\pi_\text{a} =
(\bs{\alpha}_1|\dots|\bs{\alpha}_K)$ and $\pi_\text{b} = (\bs{\beta}_1 | \dots
| \bs{\beta}_J)$ such that $\pi_\text{b}\succeq \pi_\text{a}$. Then, $
R(\pi_\text{b}) \leq R(\pi_\text{a}) $.
\end{lemma}

\begin{proof}
Let us assume that $\pi_\text{a} = (\bs{\alpha}_1|\dots|\bs{\alpha}_K\}$,
$\pi_\text{b} = (\bs{\beta}_1 | \dots | \bs{\beta}_J)$ such that
$\pi_\text{b}\succeq \pi_\text{a}$, and consider a path $\texttt{p}=
(\pi_1,\dots,\pi_L)$ in $\texttt{P}(\pi_\text{a},\pi_\text{b})$ so that
$\pi_1=\pi_\text{a}$ and $\pi_L=\pi_\text{b}$. To prove the Lemma suffices to
show that $R(\pi_{j+1}) \leq R(\pi_j)$ for $j=1,\dots,L-1$. As
$\pi_1,\dots,\pi_n$ are related by covering relationships, one just needs to
prove the inequality for two partitions such that one covers the other.

Consider $\pi_1,\pi_2\in\mathcal{P}_n$ such that $\pi_2$ covers $\pi_1$. As
both partitions differ only in one elementary refinement, let us without loss
of generality assume that the refinement is done on the last cell of $\pi_1$;
i.e. $\pi_1=(\bs{\alpha}_1|\dots|\bs{\alpha}_m)$ and
$\pi_2=(\bs{\alpha}_1|\dots|\bs{\alpha}_{m-1}|\tilde{\bs{\alpha}}_m|
\tilde{\bs{\alpha}}_{m+1})$ so that $\tilde{\bs{\alpha}}_m \cup 
\tilde{\bs{\alpha}}_{m+1} = \bs{\alpha}_m$ and $\tilde{\bs{\alpha}}_m \cap
\tilde{\bs{\alpha}}_{m+1} = \emptyset$. Then
\begin{align}
  R(\pi_1) - R(\pi_2) &= R_{\bs{\alpha}_m} - (R_{\tilde{\bs{\alpha}}_m} + R_{\tilde{\bs{\alpha}}_{m+1}} ) \nonumber \\
&= I( \bs{X}^{\tilde{\bs{\alpha}}_m}; \bs{X}^{\tilde{\bs{\alpha}}_{m+1}} | \bs{X}^{\bs{\alpha}_1}\dots \bs{X}^{\bs{\alpha}_{m-1}} ) \nonumber\\
&\geq0 \nonumber
\enspace,%
\end{align}
proving the desired result.
\end{proof}

\section{Proof of Lemma~\ref{prop:TC_DCT}  }
\label{app:TC_DCT}

\begin{proof}
Consider a path $\mathtt{p} \in \texttt{P}(\pi_\text{source},\pi_\text{sink})$,
so that $\texttt{p}= (\pi_1,\dots,\pi_L) $ with $\pi_1=\pi_\text{source}$ and
$\pi_L=\pi_\text{sink}$. Then, by using Eqs.~\eqref{eq:def_pathweight1} and
\eqref{weight_h}, a direct calculation shows that
\begin{align}
W(\texttt{p};v_h) &=  \sum_{j=1}^{L-1} \big[ H(\pi_{j+1}) - H(\pi_{j}) \big] \nonumber\\
&=  H(\pi_\text{sink}) - H(\pi_\text{source}) \nonumber\\
&= \sum_{i=1}^nH(X_i) - H(\bs{X}^n)~. \nonumber
\end{align}
Similarly, using Eqs.~\eqref{eq:def_pathweight1} and \eqref{weight_r} gives
\begin{align}
W(\texttt{p};v_\text{r}) =& \sum_{j=1}^{L-1} \left[ R(\pi_{j}) - R(\pi_{j+1}) \right] \nonumber \\
=& R(\pi_\text{source}) - R(\pi_{\text{sink}}) \nonumber\\
=& H(\bs{X}^n) - \sum_{i=1}^n H(X_i | \bs{X}^n_{-i})~. \nonumber
\end{align}
Both results make use of the fact that $W(\texttt{p};v_\text{h})$ and
$W(\texttt{p};v_\text{r})$ are telescopic sums and all but the first and last
terms cancel out.
\end{proof}

\section{Proof of Proposition~\ref{teo:1} }
\label{app:theo:1}

\begin{proof}
Let us consider a path $\mathtt{p} \in
\texttt{P}(\pi_\text{source},\pi_\text{sink})$. Then,
\begin{align}
W(\mathtt{p};v_\text{s}) &= \sum_{j=1}^L v_\text{s}(\pi_j,\pi_{j+1})  \label{eq:dec1}\\
&= \sum_{j=1}^L v_\text{h}(\pi_j,\pi_{j+1}) - \sum_{k=1}^L v_\text{r}(\pi_k,\pi_{k+1}) \nonumber\\
&= C(\bs{X}^n) - B(\bs{X}^n) = \Omega(\bs{X}^n), \nonumber
\end{align}
which proves the first part of the theorem.

Thanks to Eq.~\eqref{eq:dec1}, one can prove the second part of the Theorem by
showing that if $\pi_\text{a},\pi_\text{b}\in\mathcal{P}_n$ such that
$\pi_\text{b} \succeq \pi_\text{a}$, then $v_\text{s}(\pi_1,\pi_2)$ is equal to
an interaction information. To show this, first note that if $\pi_\text{b}
\succeq \pi_\text{a}$ then both partitions differ only in one elementary
refinement. Without no loss of generality, we assume that the refinement is
done on the last cell, such that
$\pi_\text{a}=(\bs{\alpha}_1|\dots|\bs{\alpha}_m)$ and
$\pi_\text{b}=(\bs{\alpha}_1|\dots|\bs{\alpha}_{m-1}|\tilde{\bs{\alpha}}_m|\tilde{\bs{\alpha}}_{m+1})$
such that $\tilde{\bs{\alpha}}_m \cap \tilde{\bs{\alpha}}_{m+1} = \emptyset$
and $\tilde{\bs{\alpha}}_m \cup \tilde{\bs{\alpha}}_{m+1} = \bs{\alpha}_m$.
Then,
\begin{align}
v_\text{s}(\pi_\text{a},\pi_\text{b})  &= v_\text{h}(\pi_\text{a},\pi_\text{b}) - v_\text{r}(\pi_\text{a},\pi_\text{b}) \nonumber\\
&= \big[ H(\pi_\text{b}) - H(\pi_\text{a}) \big] - \big[ R(\pi_\text{a}) - R(\pi_\text{b}) \big] \nonumber\\
&= I(\bs{X}^{\tilde{\bs{\alpha}}_m}; \bs{X}^{\tilde{\bs{\alpha}}_{m+1}}) \nonumber\\
&\quad - I(\bs{X}^{\tilde{\bs{\alpha}}_m}; \bs{X}^{\tilde{\bs{\alpha}}_{m+1}} | \bs{X}^{\bs{\alpha}_1}\dots \bs{X}^{\bs{\alpha}_{m-1}}) \nonumber\\
&=  I(\bs{X}^{\tilde{\bs{\alpha}}_m}; \bs{X}^{\tilde{\bs{\alpha}}_{m+1}} ; \bs{X}^{\bs{\alpha}_1}\dots \bs{X}^{\bs{\alpha}_{m-1}}) \nonumber
\enspace,
\end{align}
which proves the desired result.
\end{proof}

\section{Proof of Lemma~\ref{prop:bounds}}
\label{app:bounds}

\begin{proof}
Let us first note that
\begin{align}
\log |\mathcal{X}| &\geq I(X_i;X_j|X_k) \geq 0\enspace,\label{bounds1}\\
\log |\mathcal{X}| &\geq I(X_i;X_j;X_k) 
\geq - \log |\mathcal{X}| \enspace, \label{bounds2}
\end{align}
for all $i,j,k \in \{1,\dots,n\}$. Above, Eq.~\eqref{bounds2} follows from
noting that $I(X_i;X_j;X_k) = I(X_i;X_j) - I(X_i;X_j|X_k)$, and applying the
bounds in Eq.~\eqref{bounds1}. The proposition is proved by applying these
inequalities on Eqs.~\eqref{bound_TC}, \eqref{bound_DTC}, \eqref{bound_o-info},
and \eqref{eq:sumCB}. Finally, the tightness of the bounds is a direct
consequence of the tightness of Eqs.~\eqref{bounds1} and \eqref{bounds2}.
\end{proof}

\section{Proof of Proposition~\ref{prop:iff}}
\label{app:iff}

\begin{proof}
Let us first prove the first statement. By considering $\bs{X}^n$ to be a
$n$-bit copy, a direct calculation using Eqs.~\eqref{bound_TC} and
\eqref{bound_DTC} shows that $C(\bs{X}^n) = n-1$ and $B(\bs{X}^n) = 1$, and
therefore the upper bound is attained. To prove the converse, let us start by
assuming that $\Omega(\bs{X}^n)=n-2$. By applying \eqref{bounds2} to each term
in \eqref{bound_o-info}, is clear that $I(X_j;\bs{X}^{j-1};\bs{X}_{j+1}^n) = 1$
holds for all $j\in\{1,\dots,n\}$. In particular $I(X_2;X_1;\bs{X}_3^n)=1$
holds, which due to Eq.~\eqref{bound_o-info} implies that $I(X_2;X_1|\bs{X}_3^n) =
0$ and hence $I(X_2;X_1)=1$, which in turns implies that $X_1$ and $X_2$ are
Bernoulli distributed with parameter $p=1/2$, and also that $X_1=X_2$. By
relabelling the variables and following the same rationale one can prove that
every pair of variables are equal to each other, which proves that $\bs{X}^n$
is a $n$-bit copy.

Let us prove the second statement. By considering now $\bs{X}^n$ to be a
$n$-bit \texttt{xor}, using Eqs.~\eqref{bound_TC} and \eqref{bound_DTC} it is
direct to check that $C(\bs{X}^n) = 1$ and $B(\bs{X}^n) = n-1$, and hence the
lower bound is attained. To prove the converse, let us assume that $\bs{X}^n$
is such that $\Omega(\bs{X}^n)=2-n$. By considering the bounds given by
Eq.~\eqref{bounds2} in Eq.~\eqref{bound_o-info}, this implies that
$I(X_j;\bs{X}^{j-1};\bs{X}_{j+1}^n) = -1$ for all $j\in\{2,\dots,n-1\}$, and in
particular $I(\bs{X}^{n-2} ; X_{n-1};X_n) = -1$. Due to Eq.~\eqref{bounds2}, this
implies in turn that $I(\bs{X}^{n-2};X_{n-1})=0$, and via relabeling one can
prove that $\bs{X}^{n-1}$ are jointly independent. Moreover, $I(\bs{X}^{n-2} ;
X_{n-1};X_n) = -1$ also implies that $I(X_{n-1};X_n|\bs{X}^{n-2})=1$, which
implies that
\begin{equation}
I(\bs{X}^{n-1} ; X_n) = I(X_{n-1};X_n|\bs{X}^{n-2}) + I(\bs{X}^{n-2};X_n)= 1. \nonumber
\end{equation}
This equality implies that $X_n$ is Bernoulli distributed with $p=1/2$, and
that $X_n$ is a deterministic function of $\bs{X}^{n-1}$. Moreover, the fact
that $I(X_1;X_n|\bs{X}_2^{n-1})=1$ implies that for given $\bs{X}_2^{n-1}$ then
$X_n$ is a function of $X_1$, while via relabelling one finds that
$I(X_1;X_n)=0$. Since the only functions with these properties are functions
isomorphic to an $n$-variate \texttt{xor}, this proves the desired result.
\end{proof}

\section{Proof of Proposition~\ref{theo:bounds}}\label{app:bounds}

The following proof uses Lemma~\ref{lemma:bounds_C}, which is stated and proved
afterwards in this Appendix.

\begin{proof}
To prove Eq.~\eqref{eq:11}, first note that
\begin{equation}
\Omega(\bs{X}^n) = C(\bs{X}^{n-1}) - B(\bs{X}^{n-1}|X_n)  \leq C(\bs{X}^{n-1})~. \nonumber
\end{equation}
Then, the inequality follows form a direct application of
Lemma~\ref{lemma:bounds_C}. As $C(\bs{X}^m) \geq 0$, the equality becomes
non-trivial when
\begin{equation}
\Omega(\bs{X}^n) - (n-m-1) \log |\mathcal{X}| \geq  0 \nonumber
\enspace.
\end{equation}

To prove Eq.~\eqref{eq:12}, note that by using Eqs.~\eqref{bound_TC},
\eqref{bound_DTC}, and \eqref{bound_o-info} one can find that
\begin{align}
\Omega(\bs{X}^n) = & C(\bs{X}^m) - B(\bs{X}^m|\bs{X}_{m+1}^n) \nonumber \\
&+ \sum_{j=m+1}^{n-1} I(X_j ; \bs{X}^{j-1} ; \bs{X}_{j+1}^n) \nonumber \\
\geq & C(\bs{X}^m) - (n-2) \log |\mathcal{X}|. \nonumber
\end{align}
Above, the inequality is due to $I(X_j ; \bs{X}^{j-1} ; \bs{X}_{j+1}^n) \leq
\log |\mathcal{X}|$ and $B(\bs{X}^m|\bs{X}_{m+1}^n) \leq (m-1) \log
|\mathcal{X}|$. As the above relationship does not depend on the labelling of
the $X$'s, this proves Eq.~\eqref{eq:12}. As $C(\bs{X}^m) \leq
(m-1)\log|\mathcal{X}|$, the equality becomes non-trivial when
\begin{equation}
\Omega(\bs{X}^n) + (n-2) \log |\mathcal{X}| \leq  (m-1)\log|\mathcal{X}| \nonumber
\enspace.
\end{equation}
\end{proof}

\begin{lemma}\label{lemma:bounds_C}
If $|\mathcal{X}|=\min_{i=1,\dots,n}|\mathcal{X}_i|$, then
\begin{equation}
\min_{|\bs{\gamma}|=m}C(\bs{X}^\gamma) \geq C(\bs{X}^n) - (n-m)\log|\mathcal{X}|~. \nonumber
\end{equation}
\end{lemma}
\begin{proof}
A direct calculation using Eq.~\eqref{bound_TC} shows that
\begin{align}
C(\bs{X}^n) &= C(\bs{X}^m) + \sum_{j=m+1}^n I(X_j;\bs{X}^{j-1}) \nonumber\\
&\leq C(\bs{X}^m) + (n-m)\log|\mathcal{X}|. \nonumber
\end{align}
As the labelling of the indices can be modified without changing this result,
this suffices to prove the desired result.
\end{proof}

\bibliographystyle{IEEEtran}
\bibliography{Rosas}

\end{document}